\documentclass{scrartcl}
\pdfoutput=1
\usepackage{enumerate}
\usepackage{url}
\usepackage[utf8]{inputenc}
\usepackage[T1]{fontenc}
\usepackage[english]{babel}
\usepackage[dvipsnames]{xcolor}
\usepackage{amssymb,amsmath,amsfonts,bbm}
\usepackage{stmaryrd}
\SetSymbolFont{stmry}{bold}{U}{stmry}{m}{n}
\usepackage{anyfontsize}
\usepackage{amsthm}
\usepackage[labelfont={sf,bf}]{caption}
\usepackage{xspace}
\usepackage{algorithm}
\usepackage{algorithmic}

\usepackage[
  separate-uncertainty = true,
  multi-part-units = repeat
]{siunitx} 

\usepackage{tikz}
\usetikzlibrary{arrows,automata,trees,backgrounds,decorations.pathmorphing,positioning,calc,shapes.geometric}
\tikzset{shorten >=1pt, >=stealth, auto, node distance=40, initial text=}

\definecolor{darkgray}{rgb}{0.33, 0.33, 0.33}
\definecolor{lightgray}{rgb}{0.6, 0.6, 0.6}
\definecolor{myred}{RGB}{255,15,0}
\definecolor{ao(english)}{rgb}{0.0, 0.5, 0.0} 

 \usepackage{hyperref}
 \hypersetup{pdfauthor={Emmanuel Filiot and Sarah Winter},pdftitle={Synthesizing Computable Functions from Rational Specifications over Infinite Words},hidelinks,colorlinks=false,linkcolor={myred},citecolor={ao(english)}}
\usepackage[capitalize,nameinlink]{cleveref} 
\crefname{desc}{Item}{Items}
\crefname{line}{Line}{Lines}
\crefname{algo}{Algorithm}{Algorithms}
\crefname{prop}{Proposition}{Propositions}
\crefname{cor}{Corollary}{Corollaries}
\crefname{lem}{Lemma}{Lemmas}
\crefname{rem}{Remark}{Remarks}
\crefname{ex}{Example}{Examples}
\crefname{def}{Definition}{Definitions}

\makeatletter
\def\th@plain{%
  \thm@headfont{\bfseries\sffamily}
  \thm@notefont{\normalfont\sffamily}
  \itshape 
}
\def\th@definition{%
  \thm@headfont{\bfseries\sffamily}
  \thm@notefont{\normalfont\sffamily}
  \normalfont 
}
\makeatother

\theoremstyle{plain}
\newtheorem{theorem}{Theorem}
\newtheorem{lemma}[theorem]{Lemma}
\newtheorem{corollary}[theorem]{Corollary}
\newtheorem{proposition}[theorem]{Proposition}

\theoremstyle{definition}
\newtheorem{definition}[theorem]{Definition}

\newtheorem{example}[theorem]{Example}
 \newtheorem{remark}[theorem]{Remark}

\theoremstyle{remark}

\makeatletter
\newcommand{\charfusion}[2]{%
  \def\ch@rfusion##1##2{%
    \ooalign{\hfil$##1#1$\hfil\cr\hfil$##2$\hfil\crcr}}%
      \mathop{%
      \vphantom{#1}%
      \mathpalette\ch@rfusion#2}\displaylimits}
\makeatother

\newcommand{\prefs}[1]{\mathrm{Prefs}(#1)}
\newcommand{\eve}{{\mathsf{\exists}}}
\newcommand{\adam}{{\mathsf{\forall}}}

\newcommand{\RAT}{\ensuremath{\mathrm{RAT}}\xspace}
\newcommand{\DRAT}{\ensuremath{\mathrm{DRAT}}\xspace}

\newcommand{\AUT}{\ensuremath{\mathrm{AUT}}\xspace}

\newcommand{\modif}[1]{#1}

\begin{document}

\title{Synthesizing Computable Functions from\\ \mbox{Rational Specifications over Infinite Words}}

\author{Emmanuel Filiot \qquad Sarah Winter\\
Université libre de Bruxelles, Belgium}

\maketitle

\begin{abstract}
  The synthesis problem asks to automatically generate, if it exists, an algorithm
  from a specification of correct input-output pairs. In this paper, we consider the synthesis of computable
  functions of infinite words, for a classical Turing computability
  notion over infinite inputs. We consider specifications which are rational relations of infinite words, i.e., specifications defined by non-deterministic
  parity transducers. We prove that the synthesis problem of
  computable functions from rational specifications is
  undecidable. We provide an incomplete but sound reduction to some
  parity game, such that if Eve wins the game, then the
  rational specification is realizable by a computable function. We
  prove that this function is even computable by a deterministic two-way
  transducer.

  We provide a sufficient condition under which the latter
  game reduction is complete. This entails the decidability of the
  synthesis problem of computable functions, which we proved to be \textsc{ExpTime}-complete, for a large subclass of
  rational specifications, namely deterministic rational
  specifications.  This subclass contains the class of automatic
  relations over infinite words, a yardstick in reactive synthesis.
\end{abstract}

\section{Introduction}
\label{sec:intro}

Program synthesis aims at automatically generating programs from
specifications. This problem can be formalized as follows. There are four parameters: two sets of input and
output domains $I,O$, a class $\mathcal{S}$ of relations (called
specifications) from $I$ to $O$, and a class $\mathcal{I}$ of (partial) functions
(called implementations) from $I$ to $O$. Then, given a
specification $S\in \mathcal{S}$ defining the correct
input/output relationships, the synthesis problem asks to check
whether there exists a function $f\in \mathcal{I}$ satisfying $S$ in
the following sense: its graph is included in $S$ and it has the same
domain as $S$ (i.e., $f$ is defined on $x\in I$ iff $(x,y)\in S$ for
some $y\in O$). Using a set-theoretic terminology, $f$ is said to
\emph{uniformize} $S$. Moreover in synthesis, if such an $f$ exists,
then the synthesis algorithm should return (a finite presentation of)
it.

Program synthesis quickly turns to be undecidable depending on the
four parameters mentioned before. Therefore, research on synthesis either
turn to developing efficient sound but incomplete methods, see for example
the syntax-guided synthesis approach~\cite{DBLP:series/natosec/AlurBDF0JKMMRSSSSTU15} or bounded
synthesis~\cite{DBLP:journals/sttt/FinkbeinerS13,DBLP:conf/atva/GerstackerKF18}, or restrict the class of specifications
$\mathcal{S}$ and/or the class of implementations $\mathcal{I}$. A
well-known example of the latter approach is reactive synthesis, where
$\mathcal{S}$ are automatic relations\footnote{relations recognized by
two-tape parity automata alternatively reading one input and
one output symbol.} over infinite words, and
$\mathcal{I}$ are Mealy machines~\cite{BuLa69,PnuRos:89,clarke2018handbook}. Infinite words (over a
finite alphabet) are used to
model infinite executions of reactive systems, and Mealy machines are
used as a model of reactive systems processing bit streams.

In this paper, our goal is to synthesize, from specifications which
are semantically binary relations of infinite words, stream-processing
programs, which are semantically \emph{streaming computable} functions
of infinite words (just called \emph{computable} functions in the
sequel). Let us now make the computability notion we use more precise. Let $\Sigma$ and
$\Gamma$ be two finite alphabets. A partial function $f\colon \Sigma^\omega \to \Gamma^\omega$, whose
domain is denoted $\mathrm{dom}(f)$, is said to be computable, if there exists a deterministic (Turing) machine $M$ with three tapes, a read-only one-way input tape, a two-way working tape, and a write-only output tape that works as follows:
if the input tape holds an input sequence $\alpha \in \mathrm{dom}(f)$, then $M$ outputs longer and longer prefixes of $f(\alpha)$ when reading longer and longer prefixes of $\alpha$.
A definition of this machine model can be found, for instance, in
\cite{weihrauch2012computable}.

\begin{example}\label[ex]{ex:cont}
Over the alphabet $\Sigma = \Gamma = \{a,b,A,B\}$, consider the
specification given by the relation $R_1 = \{(ux\alpha, xu\beta)\mid u\alpha,u\beta\in\{a,b\}^\omega,x\in \{A,B\} \}$.
The relation $R_1$ is automatic: an automaton needs to check that
the input prefix $u$ occurs shifted by one position on the output,
which is doable using only finite memory. Checking
that the first output letter $x$ also  appears after $u$  on the input
can also be done by storing $x$ in
the state. Note that some acceptance condition (e.g., parity)  is
needed to make sure that $x$ is met again on the input. There is no
Mealy machine which can realize $R_1$, because Mealy machines operate in a
synchronous manner: they read one input symbol and must
deterministically produce one output symbol.  Here, the first output symbol which has to be
produced depends on the letter $x$ which might appear arbitrarily far in
the input sequence. However, $R_1$ can be uniformized by a computable
function: there is an \emph{algorithm} reading the input from left to right and
which simply waits till the
special symbol $x\in\{A,B\}$ is met on the input. Meanwhile, it stores
longer and longer prefixes of $u$ in memory (so it needs unbounded
memory) and once $x$ is met, it outputs $xu$. Then, whatever it reads on the input, it just copies it on the
output (realizing the identity function over the remaining infinite
suffix $\alpha$). Note that this algorithm produces a correct output
under the assumption that $x$ is eventually read. 
\end{example}

\subparagraph{Contributions.}  We first
investigate the synthesis of computable functions from \emph{rational
  specifications}, which are those relations recognizable by
non-deterministic finite state transducers, i.e., parity automata over a
product of two free monoids. We however show this problem is
undecidable (\cref{thm:undecidable}). We then give an incomplete but sound
algorithm in \cref{sec:game}, based on a reduction to $\omega$-regular two-player games. Given a transducer
$\mathcal T$ defining a specification $\mathcal R_\mathcal T$, we show how to effectively
construct a two-player game $\mathcal G_\mathcal T$, proven to be solvable in
\textsc{ExpTime}, such that if Eve wins $\mathcal G_\mathcal T$, then there
exists a computable function which uniformizes the relation recognized
by $\mathcal T$, which can even be computed by some input-deterministic \textbf{two-way} finite state transducer (a transducer which whenever it
reads an input symbol deterministically produces none or several
output symbols and either moves forward or backward on the input). It
is easily seen that two-wayness is necessary: the relation $R_1$
of Example~\ref{ex:cont} cannot be uniformized by any deterministic device which moves only
forward over the input and only uses finitely many states, as the
whole prefix $u$ has to be remembered before reaching $x$. However, a
two-way finite-state device can do it: first, it scans the prefix up to $x$, comes
back to the beginning of the input, knowing whether $x=A$ or $x=B$, and then can produce the output.

Intuitively, in the two-player game we construct, called
unbounded delay game, Adam picks the input symbols while Eve picks the
output symbols. Eve is allowed to delay her moves an arbitrarily
number of steps, gaining some lookahead information on Adam's
input. We use a finite abstraction to store the lookahead gained on
Adam's moves. We show that any finite-memory winning strategy in this game can be
translated into a function uniformizing the specification such that it
is computable by an input-deterministic two-way
transducer.

In \cref{sec:completeness}, we provide a sufficient condition $\mathcal{P}$ on
relations for which the game reduction is complete. In particular, we show that if a
relation $R$ satisfies $\mathcal{P}$, then Eve wins the game iff $R$
can be uniformized by a computable function. A large subclass of
rational relations satisfying this sufficient condition is the class
of deterministic rational relations (\DRAT, see~\cite{sakarovitch2009elements}).
Deterministic rational relations are
those relations recognizable by deterministic two-tape automata, one
tape holding the input word while the other holds the output
word. It strictly subsumes the class of automatic relations, and,
unlike for automatic relations, the two heads are not required to
move at the same speed.
Furthermore, when the domain of the relation is topologically closed for the Cantor
distance\footnote{\label{note}A set $X\subseteq\Sigma^\omega$ is \emph{closed} if
the limit, if it exists, of any sequence $(x_i)_i$ of infinite words
in $X$ is in $X$. The limit here is defined based on the Cantor
distance, which, for any two infinite words $u,v$, is $0$ if $u=v$ and otherwise $2^{-\ell}$ where $\ell$ is the
length of their longest common prefix.}, we show that strategies in which Eve delays her moves at most a bounded number of steps are sufficient for Eve to win. Such a strategy can in turn be converted into an input-deterministic \textbf{one-way} transducer. This entails that for \DRAT-specifications with a closed domain (such as for instance specifications with domain $\Sigma^\omega$, i.e., total domains), if it is uniformizable by a computable function, then it is uniformizable by a function computable by an input-deterministic one-way transducer.

Based on the completeness result, we prove
our main result, that
the synthesis problem of computable
functions from deterministic rational relations is
\textsc{ExpTime}-complete. Hardness also holds in the particular case of
automatic relations of total domain.

\subparagraph{Total versus partial domains.} We would like to emphasize
here on a subtle difference between our formulation of synthesis
problems and the classical formulation in reactive synthesis.
Classically in reactive synthesis, it is required that a controller produces
for \emph{every} input sequence an output sequence correct w.r.t.\ the
specification. Consequently, specifications with partial
domain are by default unrealizable. So, in this setting, the specification $R_1$ of
the latter example is not realizable, simply because its domain is not
total (words with none or at least two occurrences of a symbol in
$\{A,B\}$ are not in its domain). In our definition, specifications
with partial domain can still be realizable, because the synthesized function, if it exists, can be partial and must be
defined only on inputs for which there exists at least one matching
output in the specification. A well-known notion corresponding to this
weaker definition is that of
uniformization~\cite{IC::Kobayashi69,choffrut1999uniformization,CarayolL14,DBLP:conf/icalp/FiliotJLW16},
this is why we often use the terminology ``uniformizes'' instead
of the more widely used terminology ``realizes''. The problem of
synthesizing functions which uniformize quantitative
specifications has been recently investigated in~\cite{almagor2020good}. In~\cite{almagor2020good},
it was called the good-enough synthesis problem, a controller
being good-enough if it is able to compute outputs for all
inputs for which there exists a least one matching output by the
specification. The uniformization setting allows to formulate assumptions
that the input the program receives is not any input, but belongs to
some given language. Related to that, there is a number of works on reactive synthesis under
assumptions on the behavior of the
environment~\cite{DBLP:conf/tacas/ChatterjeeH07,DBLP:conf/atva/BloemEK15,DBLP:journals/amai/KupfermanPV16,DBLP:conf/icalp/ConduracheFGR16,DBLP:journals/acta/BrenguierRS17}.

\subparagraph{Related work.}\label{sec:related}
To the best of our knowledge, this work is the first
contribution which addresses the synthesis of algorithms from
specifications which are relations over infinite inputs, and such that
these algorithms may need unbounded memory, as illustrated by the
specification $R_1$ for which any infinite-input Turing-machine realizing the
specification needs unbounded memory. There are however two related
works, in some particular or different settings.

First, in~\cite{DBLP:conf/concur/DaveFKL20}, the synthesis of computable
functions has been shown to be decidable in the
particular case of \emph{functional} relations, i.e., graphs of
functions. The main contribution
of~\cite{DBLP:conf/concur/DaveFKL20} is to prove that checking whether
a function represented by a non-deterministic two-way transducer is
computable is decidable, and that computability coincides with
continuity (for the Cantor distance) for this large class of functions. The techniques
of~\cite{DBLP:conf/concur/DaveFKL20} are different to ours, e.g., games are not
needed because output symbols functionally depends on input ones, even in the
specification, so, there are no choices to be made by Eve.

Second, Hosch and Landweber \cite{hosch1972finite} proved decidability
of the synthesis problem of Lipschitz-continuous functions from
automatic relations with \emph{total} domain, Holtmann, Kaiser and Thomas \cite{holtmann2010degrees} proved decidability of the synthesis
of continuous functions from automatic relations with
\emph{total} domain, and Klein and Zimmermann
\cite{DBLP:journals/corr/KleinZ14} proved
\textsc{ExpTime}-completeness for the former problem. So, we inherit the lower bound because
automatic relations are particular \DRAT relations, and as
we show in the last section of the paper, the synthesis problem of computable functions coincides with
the synthesis problem of continuous functions. We obtain the same
upper bound as~\cite{DBLP:journals/corr/KleinZ14} for a more general
class of specifications, namely \DRAT, and in the more general setting
of specifications with partial domain. As we show, total vs.\ partial
domains make an important difference: two-way transducers may be
necessary in the former case, while one-way transducers are sufficient
in the latter.  \cite{holtmann2010degrees,DBLP:journals/corr/KleinZ14}
also rely on a reduction to two-player games called delay games, but for which bounded delay
are sufficient. However, our game is built such that it accounts for the fact that
unbounded delays can be necessary and it also monitors the domain, which is
not necessary
in~\cite{holtmann2010degrees,DBLP:journals/corr/KleinZ14} because
specifications have total domain. Accordingly, the main
differences between \cite{DBLP:journals/corr/KleinZ14} and our delay
games are their respective winning objectives and correctness
proofs. Another difference is that our game applies to the general
class of rational relations, which are \emph{asynchronous} (several
symbols, or none, can correspond to a single input symbol) in contrast
to automatic relations which are synchronous by definition.

The work is an extended version of the conference work \cite{DBLP:conf/fsttcs/FiliotW21} with additional proof details.
\section{Preliminaries}
\label{sec:prelims}

\subparagraph{Words, languages, and relations.}

Let $\mathbbm N$ denote the set of non-negative integers. Let $\Sigma$
and $\Gamma$ denote \emph{alphabets} of elements called \emph{letters} or \emph{symbols}.
A \emph{word} resp.\ \emph{$\omega$-word} over $\Sigma$ is an empty or
non-empty finite resp.\ infinite sequence of letters over $\Sigma$.
Let $\Sigma^*$, $\Sigma^+$, and $\Sigma^\omega$ denote the set of finite, non-empty finite, and infinite words over $\Sigma$, respectively.
Let $\Sigma^\infty$ denote $\Sigma^* \cup \Sigma^\omega$.
The empty word is denoted by $\varepsilon$, the length of a word by $|\cdot|$.
Usually, we denote finite words by $u,v,w$, etc., and infinite words by $\alpha,\beta,\gamma$, etc.
Given an (in)finite word $\alpha = a_0a_1\cdots$ over $\Sigma$ with
$a_0,a_1,\dots \in \Sigma$, let $\alpha(i)$ denote the letter $a_i$,
$\alpha(i\colon\!j)$ denote the infix $a_ia_{i+1}\cdots a_j$,
$\alpha(\colon\!i)$ the prefix $a_0a_1\cdots a_i$, and
$\alpha(i{+}1\colon\!)$ the suffix $a_{i+1}\cdots$ for $i \leq j \in
\mathbbm N$. \modif{Given two words $u\in\Sigma^*$ and
  $\alpha\in\Sigma^\infty$, we say that $u$ is a prefix of $\alpha$, denoted
  $u\preceq \alpha$, if $\alpha(\colon\!i) = u$ for some
  $i<|\alpha|$, and in that case let $u^{-1}\alpha = \alpha(i\colon\!)$. We also denote by $\prefs{\alpha}$ the set of all
  finite prefixes of $\alpha$, i.e., $\prefs{\alpha} = \{\alpha(\colon\!i)\mid
i<|\alpha|\}$.}
For two (in)finite words $\alpha,\beta$, let $\alpha \wedge \beta$
denote their longest common prefix, i.e., \modif{the longest word in
$\prefs{\alpha}\cap \prefs{\beta}$ if $\alpha\neq \beta$, otherwise
$\alpha$ if $\alpha=\beta$.}

A \emph{language} resp.\ \emph{$\omega$-language} $L$ is a subset of
$\Sigma^*$ resp.\ $\Sigma^\omega$, its set of finite prefixes is denoted by $\prefs{L}$.
A (binary) \emph{relation} resp.\ \emph{$\omega$-relation} $R$ is a
subset of $\Sigma^* \times \Gamma^*$ resp.\ $\Sigma^\omega \times
\Gamma^\omega$. An $\omega$-relation is just called a relation when
infiniteness is clear from the context. The domain $\mathrm{dom}(R)$
of a ($\omega$)-relation $R$ is the set $\{ \alpha \mid \exists \beta\
(\alpha,\beta) \in R\}$. It is \emph{total} if $\mathrm{dom}(R) =
\Sigma^*$ resp.\ $\Sigma^\omega$. Likewise, we define $\mathrm{img}(R)$ the image of $R$, as
the domain of its inverse.
A relation $R$ is \emph{functional} if for each $u \in
\mathrm{dom}(R)$ there is at exactly one $v$ such that $(u,v) \in R$.
\emph{By default in this paper, relations and functions are
partial, i.e., are not necessarily total.} 
\modif{Given a function $f\colon \Sigma^\omega \to \Gamma^\omega$, let $\hat f\colon \Sigma^* \to \Gamma^\infty$ denote the function defined for all $u \in \prefs{\mathrm{dom}(f)}$ as
\begin{equation}\label{eq:fhat}
\hat f(u) = \bigwedge \{ f(u\alpha) \mid  u\alpha \in \mathrm{dom}(f) \},
\end{equation}
that is, the longest common prefix of all outputs of $f$ for inputs that begin with $u$.
}

\subparagraph{Automata.}

A \emph{parity automaton} is a tuple $\mathcal A = (Q,\Sigma,q_0,\Delta,c)$, where $Q$ is a finite set of states, $\Sigma$ a finite alphabet, $q_0 \in Q$ an initial state, $\Delta \subseteq Q \times \Sigma \times Q$ a transition relation, and $c\colon Q \to \mathbbm N$ is a function that maps states to \emph{priorities}, also called \emph{colors}.
A parity automaton is \emph{deterministic} if its transition relation $\Delta$ is given as a transition function $\delta$. We denote by $\delta^*$ the usual extension of $\delta$ from letters to finite words.
\modif{A \emph{run of $\mathcal A$ on a word $w \in \Sigma^\infty$} is a word
$\rho \in Q^\infty$ of length $|w|+1$ if $w$ is finite, otherwise of
infinite length,  such that $(\rho(i),w(i),\rho(i+1)) \in \Delta$ for all $0 \leq i < |w|$.}
A run on $\varepsilon$ is a single state.
We say that $\rho$ begins in $\rho(0)$ and ends in $\rho(|w|)$ if $w$ is finite.

\modif{Given a set $E$ and a sequence $s\in E^\infty$, we let
$\mathrm{Occ}(s)$ as the set of elements of $E$ that occur in $s$,
and $\mathrm{Inf}(s)$ as the set of elements of $E$ that occur
infinitely often in $s$. In particular $\mathrm{Inf}(s)=\varnothing$
if $s$ is finite.} Given a run $\rho$, we let
$c(\rho)=c(\rho(0))c(\rho(1))\cdots$ be the sequence of colors seen
along that run.  The run $\rho$ is \emph{accepting} if $\rho\in Q^\omega$, $\rho(0) =
q_0$ and $\mathrm{max}\ \mathrm{Inf}(c(\rho))$ is even.
The language \emph{recognized} by $\mathcal A$ is the set $L(\mathcal
A) = \{ \alpha \in \Sigma^\omega \mid \text{there is an accepting run $\rho$ of $\mathcal A$ on $\alpha$} \}$.
A language $L \subseteq \Sigma^\omega$ is called \emph{regular} if $L$ is recognizable by a parity automaton.



    



\subparagraph{One-way transducers.}

A \emph{transducer} (1NFT) is a tuple $\mathcal T =
(Q,\Sigma,\Gamma,q_0,\Delta,c)$, where $Q$ is finite state set,
$\Sigma$ and $\Gamma$ are finite alphabets, $q_0 \in Q$ is an initial
state, $\Delta \subseteq Q \times \Sigma^* \times \Gamma^* \times Q$
is a finite set of transitions, and $c\colon Q \to \mathbbm N$ is a
parity function.
It is \emph{input-deterministic} (1DFT) (also called \emph{sequential} in the
literature) if $\Delta$ is expressed as a function $Q
\times \Sigma \to \Gamma^* \times Q$. A \emph{finite non-empty run}
$\rho$ is a non-empty sequence of transitions of the form
$(p_0,u_0,v_0,p_1)(p_1,u_1,v_1,p_2)\cdots (p_{n-1},u_{n-1},v_{n-1},p_n)
\in \Delta^*$.
The \emph{input (resp.\ output)} of $\rho$ is $\alpha = u_0\cdots u_{n-1}$ (resp.\ $\beta = v_0\cdots v_{n-1}$).
As shorthand, we write $\mathcal T\colon p_0 \xrightarrow{\alpha/\beta} p_n$.
An \emph{empty run} is denoted as $\mathcal T\colon p \xrightarrow{\varepsilon/\varepsilon} p$ for all $p\in Q$.
Similarly, we define an \emph{infinite run}.
A run is \emph{accepting} if it is infinite, begins in the initial
state and satisfies the parity condition. In this paper, we also
assume that for any accepting run $\rho$, its input and output are
both infinite. This can be syntactically ensured with the parity
condition. The relation \emph{recognized} by $\mathcal{T}$ is
$R(\mathcal T) = \{ (\alpha,\beta) \mid \textnormal{there is an
  accepting run of $\mathcal T$ with input $\alpha$ and output
  $\beta$}\}$. Note that with the former assumption, we have
$R(\mathcal T)\subseteq \Sigma^\omega\times \Gamma^\omega$. 
A relation is called \emph{rational} if it is recognizable by a transducer, we denote by \RAT the class of rational relations.
A \emph{sequential} function is a function whose graph is $R(\mathcal T)$ for an input-deterministic transducer $\mathcal T$.

\subparagraph{Two-way transducers.}

Given $\Sigma$, let $\Sigma_\vdash$ denote $\Sigma \uplus \{\vdash\}$, $\vdash$ is a new left-delimiter symbol.
An \emph{input-deterministic two-way transducer} (2DFT) is a tuple
$\mathcal T = (Q,\Sigma_\vdash,\Gamma,q_0,\delta,c)$, where $Q$ is a
finite state set, $\Sigma$ and $\Gamma$ are finite alphabets, $q_0\in
Q$ is an initial state, $\delta\colon Q \times \Sigma_\vdash \to
\Gamma^* \times \{1,-1\} \times Q$ is a transition function \modif{such that
for all states $q$, $\delta(q,\vdash)\in \Gamma^*\times \{1\}\times Q$}, and $c\colon Q \to \mathbbm N$ is a function that maps states to colors.
A two-way transducer has a two-way read-only input tape and a one-way write-only output tape.
Given an input sequence $\alpha \in \Sigma^\omega$, let $\alpha(-1) = \ \vdash$, the input tape holds $\vdash\!\alpha$.
We denote a transition $\delta(p,a) = (\gamma,d,q)$ as a tuple $(p,a,\gamma,d,q)$, and $\Delta$ denotes the tuple representation of $\delta$.
A \emph{run} of $\mathcal T$ on $\alpha \in \Sigma^\omega$ is a sequence of transitions 
$
\rho = (q_0,\alpha(i_0),\gamma_0,d_0,q_1)(q_1,\alpha(i_1),\gamma_1,d_1,q_2)\cdots \in \Delta^\omega
$ such that $i_0 = 0$, and $i_{k+1} = i_{k} + d_k$ for all $k\in\mathbbm N$.
The \emph{input} of $\rho$ is $\alpha$ and the \emph{output} of $\rho$
is $\beta = \gamma_0\gamma_1\cdots$. We define $c(\rho)$ as the
sequence of colors $c(q_0)c(q_1)\cdots$, $\rho$ is \emph{accepting} if
$\mathrm{max}\ \mathrm{Inf}(c(\rho))$ is even and \modif{it visits all
positions of $\alpha$, i.e., $\{i_0,i_1,\dots\} = \mathbbm{N}$.}
The functional $\omega$-relation \emph{recognized} by the deterministic-two way transducer is defined as 
$
R(\mathcal T) = \{ (\alpha,\beta) \mid \textnormal{there is an accepting run of $\mathcal T$ with input $\alpha$ and output $\beta$}\}.
$

\subparagraph{Games.}

A \emph{game arena} is a tuple $G = (V_0,V_1,v_0,A,E)$, where $V = V_0 \uplus V_1$ is a set of vertices, $V_0$ belongs to Eve and $V_1$ to Adam, $v_0$ is an initial vertex, $A$ is a finite set of actions, and $E \subseteq V \times A \times V$ is a set of labeled edges such that $(v,a,v')\in E$ and $(v,a,v'')\in E$ implies that $v' = v''$ for all $v\in V$ and $a\in A$.
We assume that the arena is deadlock-free (\modif{there is always at least one outgoing
edge from any vertex}). We use letters on edges as it is more convenient to have them at hand for the proofs, it is however not necessary.
%
A \emph{play} in $G$ is an infinite sequence $v_0a_0v_1a_1\cdots$ such that $(v_i,a_i,v_{i+1}) \in E$ for all $i \in \mathbbm N$.
Note that a play is uniquely determined by its action sequence.
%
A \emph{game} is of the form $\mathcal G = (G,\mathit{Win})$, where $G$ is a game arena and $\mathit{Win} \subseteq V^\omega$ is a winning condition.
Eve wins a play $\alpha = v_0a_0v_1a_1\cdots$ if $v_0v_1\cdots \in \mathit{Win}$, otherwise Adam wins.
For ease of presentation, we also write $\alpha \in
\mathit{Win}$. \modif{A parity condition $\mathit{Win}$ is represented by a coloring
  function $c :V\rightarrow \mathbb{N}$ and consists of the set of
  sequences $v_0a_0v_1a_1\cdots$ such that $\mathrm{max}\
  \mathrm{Inf}(c(v_0)c(v_1)\cdots)$ is even. A \emph{parity game} is a
  game $\mathcal{G} = (G,\mathit{Win})$ such that $\mathit{Win}$ is a
  parity condition (given as a coloring function $c$)}.

A \emph{strategy} for Eve resp.\ Adam is a function $\sigma\ :\ (VA)^*
V_0 \to A$ resp.\ $\sigma\ :\ (VA)^* V_1 \to A$ such that for all
$x\in  (VA)^*$ and $v \in V_0$ resp. $v\in V_1$, there exists $v'\in V$ such that
$(v,\sigma(xv),v')\in E$. A play $v_0a_0v_1a_1\cdots$ is consistent with a strategy $\sigma$ for Eve resp.\ Adam if $\sigma(v_0a_0\cdots v_i) = a_{i}$ for all $i \in \mathbbm N$ with $v_i \in V_0$ resp.\ $v_i \in V_1$.
A strategy $\sigma$ for Eve is a \emph{winning strategy} if $\alpha
\in \mathit{Win}$ for all plays $\alpha$ consistent with $\sigma$.

A \emph{strategy automaton} for Eve is a tuple $\mathcal S =
(M,V,m_0,\delta,\mu)$, where $M$ is a finite set of (memory) states,
$V$ is the alphabet, $m_0$ is an initial state, $\delta\colon M \times
V \to M$ is the memory update function, and $\mu\colon M \times V_0
\to A$ is the next action function such that for all $v\in V_0$ and $m\in
M$, there is $v'\in V$ with $(v,\mu(m,v),v') \in E$. \modif{The transition
function $\delta$ is naturally extended into a function $\delta^* :
M\times V^*\rightarrow M$ over finite sequences over $V$. 
A strategy automaton $\mathcal S$ defines a strategy $\sigma_{\mathcal
  S}$ as follows: for all $x\in (VA)^*$ and $v\in V_0$,
$\sigma_{\mathcal S}(xv) = \mu(\delta^*(m_0,\pi_V(x)),v)$ where
$\pi_V(x)$ is the projection of $x$ onto $V$.}

\subparagraph{Problem statement.}
\label{sec:main}

In this section, we introduce the problem we want to solve. Let $\Sigma,\Gamma$ be two finite
alphabets. Given a relation $R \subseteq \Sigma^\omega \times
\Gamma^\omega$ and a (partial) function $f :\Sigma^\omega \rightarrow
\Gamma^\omega$, $f$ is said to \emph{uniformize} $R$ if $\mathrm{dom}(f) =
\mathrm{dom}(R)$ and $(\alpha,f(\alpha)) \in R$ for all $\alpha \in
\mathrm{dom}(R)$. We also say that $R$ is uniformizable by $f$ or that
$f$ is a uniformizer of $R$. We are interested in computable uniformizers, which
we now introduce.

\begin{definition}[\cite{weihrauch2012computable} computable functions]\label[def]{def:computable}
A function $f\colon \Sigma^\omega \to \Gamma^\omega$ is called \emph{computable} if there exists a deterministic multi-tape machine $M$ that computes $f$ in the following sense,
$M$ has a read-only one-way input tape, a two-way working tape, and a write-only one-way output tape.
All tapes are infinite to the right, finite to the left.
For any finite word $w \in \Sigma^*$, let $M(w)$ denote the
output\footnote{The finite word written on the output tape the first
  time $M$ reaches the $|w|$-th cell of the input tape ($M$ is assumed
to visit every input cell)} of
$M$ on $w$.
The function $f$ is said to be computed by $M$ if for all $\alpha \in
\mathrm{dom}(f)$ and $i\in\mathbbm N$ there exists $j\in\mathbbm N$
such that $f(\alpha)(\colon\! i)$ is a prefix of $M(\alpha(\colon\! j))$.
\end{definition}

\begin{remark}\label[rem]{remark:domain}
  Note that in the above definition, checking whether the infinite input belongs to the domain is not a requirement and should not be, because in general, it is impossible to do it reading only a finite prefix of the input.
  That is why in this definition, we assume that the input belongs to the domain of the function.
  It is a reasonable assumption. 
  For instance, the inputs may have been produced by another program (e.g., a transducer) for which one has guarantees that they belong to some well-behaved (e.g., regular) language.
\end{remark}

\begin{example}\label[ex]{ex:cont-uni}
    To begin with, consider the function $f_1\colon \{a,b,c\}^\omega
    \to \{b,c\}^\omega$ defined by $f_1(a^nba^\omega) = b^\omega$ and
    $f_1(a^nca^\omega) = c^\omega$ for all $n\in \mathbbm N_{\geq
      1}$. It is computable by a TM which on inputs of the form $a^nxa^\omega$ for
    $x\in\{b,c\}$, outputs nothing up to reading $x$, and then,
    depending on $x$, either outputs $c$ or $b$ whenever it reads an
    $a$ in the remaining suffix $a^\omega$.

    Consider the function $f_2:\{a,b\}^\omega\rightarrow \{a,b\}^\omega$ defined by $f_2(\alpha) = a^\omega$ if $\alpha$ contains infinitely
    many $a$ and $f(\alpha) = b^\omega$ otherwise for all $\alpha \in
    \{a,b\}^\omega$. It is rational but not
    computable, because to determine even the first output letter, an
    infinite lookahead is needed.
\end{example}

Let $\mathcal{S}$ be a class of relations. The
\emph{$\mathcal{S}$-synthesis problem} asks, given a relation $S\in\mathcal{S}$ (finitely represented), whether there exists a
computable function which uniformizes $S$. If such a function exists,
then the procedure must return a TM computing it. Our first result is
an undecidability result.

\vspace{4mm}

\begin{proposition}
    \label[prop]{thm:undecidable}
The \RAT-synthesis problem is undecidable, even if restricted to the
subclass of rational relations with total domain.
\end{proposition}

\begin{proof}[Proof sketch]
We sketch a reduction from Post's correspondence problem.
Let $u_1,\dots,u_n$ and $v_1,\dots,v_n$ be a PCP instance \modif{over
alphabet $\{0,1\}$}.
We construct the $\omega$-rational relation $R$ that contains pairs
$(\alpha,\beta)$ such that $\alpha = i_1\cdots i_k\alpha'$ with
$i_1\cdots i_k \in \{1,\dots,n\}^*$ and $\alpha' \in \{a,b\}^\omega$,
and $\beta = u_{i_1}\cdots u_{i_k}a^\omega$ if $\alpha'$ contains
infinitely many $a$ and \modif{otherwise if $\alpha'$ contains
  finitely many $a$, then $\beta \neq v_{i_1}\cdots
  v_{i_k}a^\omega$}. If the instance of PCP has no solution, then the function $f\colon i_1\cdots i_k\alpha' \mapsto u_{i_1}\cdots u_{i_k}a^\omega$ uniformizes $R$, because $u_{i_1}\cdots u_{i_k} \neq v_{i_1}\cdots v_{i_k}$.
The function $f$ is clearly computable.
If the PCP has a solution, no computable function uniformizes $R$.
If the integer sequence $i_1\cdots i_k$ is the solution, then $u_{i_1}\cdots u_{i_k} = v_{i_1}\cdots v_{i_k}$.
Intuitively, for an input sequence starting with the solution, no prefix of the input sequence allows to determine whether the output must begin with $u_{i_1}\cdots u_{i_k}$ or is not allowed to begin with $u_{i_1}\cdots u_{i_k}$.

The relation $R$ can be made complete by also allowing all ``invalid'' inputs together with any output, i.e., by adding all pairs $(\alpha,\beta) \in \{1,\dots,n,a,b\}^\omega \times \{1,\dots,n,a,b\}^\omega$ where the input sequence $\alpha$ is not of the form $i_1\cdots i_k\alpha'$ with $i_1\cdots i_k \in \{1,\dots,n\}^*$ and $\alpha' \in \{a,b\}^\omega$, and any output sequence $\beta$.
Any Turing machine computing $f$ can easily be adapted to verify
whether the input is valid.
\end{proof}

Next, we give a semi-decision procedure for solving the \RAT-synthesis problem which is sound but not complete.
In \cref{sec:completeness}, we introduce a sufficient condition for
completeness which yields a (sound and complete) decision procedure
for large and expressive classes of rational relations.

\section{Unbounded Delay Game}
\label{sec:game}

In this section, given a rational relation (as a transducer), we show how to construct a finite-state $\omega$-regular two-player game
called (unbounded) delay game. We prove that if Eve wins this game then
there exists a computable function which uniformizes the relation. Moreover, this function is 
computable by an input-deterministic two-way transducer.
We analyze the complexity of solving the game, which turns out to be in \textsc{ExpTime}.
Solving this game yields an incomplete, but sound, decision procedure for the \RAT-synthesis problem.

In the game, Adam provides inputs and Eve must produce
outputs such that the combination of inputs and outputs is in the
relation. However, as seen in \cref{ex:cont}, Eve might need to wait
arbitrarily long before she can safely produce output.
Hence, as the game is finite, it can not store arbitrary long
input words, and Eve's actions cannot produce arbitrarily long
words neither. Instead, we finitely abstract input and output words
using a notion we call profiles. 
Informally, a profile of an input word stores the effects of the word (together with some output word) on the states of the transducer (that specifies the relation) as well as the maximal priority seen along the induced state transformation.
Such profiles contain sufficient
information to express a winning condition which makes sure 
that given the word of input symbols provided by Adam, if Eve
would output concrete output words instead of their abstraction, she
would produce infinitely often non-empty output words 
whose concatenation, together with the input word, belongs to
the relation. 

\subparagraph{State transformation profiles.}

Let $R\subseteq \Sigma^\omega\times \Gamma^\omega$ be a rational relation given by a transducer
$\mathcal T = (Q_\mathcal T,\Sigma,\Gamma,q_0^\mathcal
T,\Delta_\mathcal T,c_\mathcal T)$, and $C_\mathcal T =
\mathrm{img}(c_\mathcal T)$ its set of used priorities.
Let $\mathcal D = (Q_\mathcal D,\Sigma,q_0^\mathcal D,\delta_\mathcal
D,c_\mathcal D)$ be a deterministic parity automaton that recognizes
$\mathrm{dom}(R)$, and $C_\mathcal D = \mathrm{img}(c_\mathcal D)$ its
set of used priorities; $\mathcal D$ can always be constructed from
$\mathcal T$ by projecting away its outputs and by determinizing the
resulting automaton.

Given $u\in\Sigma^*$, its \emph{profile} $P_u$ are all the possible state
transformations it induces for any output. Formally, $P_u \subseteq Q_\mathcal T \times Q_\mathcal T \times C_\mathcal T$ is defined as
$
\bigl\{
(p,q,c) \mid \text{ there is } v\in\Gamma^* \text{ and there is a run } \rho \text{ of the form } \mathcal T\colon p \xrightarrow{u / v} q \text{ with } \mathrm{max\ Occ}(\rho) = c
\bigr\}
$. Profiles can be multiplied as $P_1
\otimes P_2 = \{(p,r,\mathrm{max}\{m,n\})\mid \exists q\colon (p,q,m) \in P_1, (q,r,n) \in P_2\}$.
Given $u_1, u_2 \in \Sigma^*$, it is easy to verify that $P_{u_1u_2} =
P_{u_1} \otimes P_{u_2}$, and $P_\varepsilon$ is neutral for~$\otimes$.

\subparagraph{Finite-state unbounded delay game.} We now present a
two-player $\omega$-regular game $\mathcal G_\mathcal T = (G,\mathit{Win})$
such that if Eve has a winning strategy, then $R$ has a computable uniformizer. In this game,
Adam's actions are to provide input letters, letter-by-letter. Eve's
goal is to construct a sequence of state transformations
$(q_0,q_1,m_1)(q_1,q_2,m_2)\cdots$ such that if the infinite input $\alpha \in \Sigma^\omega$ provided by Adam is in $\mathrm{dom}(R)$, then $(i)$ the maximal priority seen infinitely often in $(m_i)_i$ is even and $(ii)$ $\alpha = u_0u_1\cdots$ for some $u_i \in \Sigma^*$ such that $(q_i,q_{i+1},m_{i+1}) \in P_{u_i}$ for all $i \geq 0$.
As a consequence, all
these finite runs can be concatenated to form an accepting run on
$\alpha/v_0v_1\cdots$, entailing $(\alpha,v_0v_1\cdots)\in R$. One can
then show that if Eve has a strategy to pick the state transformations
while ensuring the latter property, then this strategy can be turned into a
computable function, and conversely. Picking a state transformation is
what we call a \emph{producing action} for Eve. Since a state transformation picked by
Eve may correspond to an arbitrarily long word $u_i$, she also has an
action $skip$ which allows her to wait before making such a producing action. Now, the
difficulty for Eve is to decide when she makes producing actions, in
other words, how to decompose the input $\alpha$, only based on
prefixes of $\alpha$. To that end, before picking a state
transformation, she may need to gather lookahead information from
Adam. Consequently, the vertices of the game manipulate two
consecutive profiles $P_1$ and $P_2$, with the invariant that $P_1$ is
the profile of $u_i$ while $P_2$ is the profile of $u_{i+1}$, when the
input played so far by Adam is $u_0\cdots u_{i+1}$. When Eve knows
enough, she picks a state transformation $(q_i,q_{i+1},m_i)$ in $P_1$,
then $P_1$ becomes $P_2$ and $P_2$ is reset to $P_\varepsilon$. The
inputs of Adam up to the next producing action of Eve form the word
$u_{i+2}$, and so on. The vertices of the game also
store information to decide whether the input belongs to the domain of
$R$ (states of $\mathcal{D}$), the parities $m_i$, as well as the states
$q_0,q_1,\dots$. Formally, the game graph $G = (V,E)$ is composed of vertices of the form
$\bigl( q,c,P_1,P_2,r \bigr) \times \{\adam,\eve\}$\label{page:game}, where
\begin{itemize}
  \item $q \in Q_\mathcal T$, \quad\emph{State reached on the combination of input and output sequence.}
  \item $c \in \{-1\} \cup C_\mathcal T$, \quad\emph{Priorities of the
      state transformations, $-1$ is used to indicate that no state
      transformation was chosen (skip action below).}
  \item $P_1,P_2$, \quad\emph{Profiles obtained from the given lookahead of the input word.}
  \item $r \in Q_\mathcal D$. \quad\emph{State reached on the given lookahead of the input word.}
\end{itemize}
From a vertex of the form $\bigl( q,c,P_1,P_2,r,\adam \bigr)$, Adam has the following actions:
\begin{itemize}
  \item $\xrightarrow{a} \bigl( q,-1,P_1,P_2 \otimes P_a,\delta_\mathcal D(r,a),\eve\bigr)$, for all $a \in \Sigma$.

\quad \emph{Adam provides the next lookahead letter and $P_2$ is
  updated accordingly.}

\end{itemize}
From a vertex of the form $\bigl( q,c,P_1,P_2,r,\eve \bigr)$, Eve has the following actions:
\begin{itemize}
  \item $\xrightarrow{skip} \bigl( q,-1,P_1,P_2,r,\adam \bigr)$, and

\quad \emph{Eve makes a non-producing action, i.e., she waits for further lookahead on the input.}

\item $\xrightarrow{(q,q',c')} \bigl( q',c',P_2,P_\varepsilon,r,\adam \bigr)$, where $(q,q',c') \in P_1$.

\quad \emph{Eve makes a producing action: a state transformation from the first lookahead profile is chosen, the state transformation is applied, and the first profile is consumed.
}
\end{itemize}
The initial vertex of the game is $\bigl( q_0^\mathcal
T,-1,P_\varepsilon,P_{\varepsilon},q_0^\mathcal D,\adam \bigr)$.

Let us now define $\mathit{Win} \subseteq V^\omega$. The condition
makes sure that if the input sequence provided by Adam is in the
domain of $R$, then the sequence of state transformations can be used
to build on accepting run of $\mathcal{T}$ on that input. 
$\mathit{Win} \subseteq V^\omega$ is the set of all plays $\gamma$ satisfying the property
\begin{equation*}
  \mathrm{max}\ \mathrm{Inf}(\mathit{col}_\mathcal D (\gamma)) \text{ is even} \rightarrow \mathrm{max}\ \mathrm{Inf}(\mathit{col}_\mathcal T (\gamma)) \text{ is even},
\end{equation*}
where $\mathit{col}_\mathcal D (\gamma) = c_\mathcal D(\pi^5(\gamma))$, $\mathit{col}_\mathcal T (\gamma) = \pi^2(\gamma)$, and 
$\pi^i(\gamma)$ is the projection of $\gamma$ onto the $i$-th component of each vertex.
It is not difficult to see that $\mathit{Win}$ is $\omega$-regular, e.g., one can design a parity automaton for it.

We explain the intuition behind $\mathit{Win}$.
Our goal is to extract a computable function that uniformizes the relation from a winning strategy.
Intuitively, there is a computable function that uniformizes $R$, if every input word $\alpha \in \mathrm{dom}(R)$ can be read letter-by-letter, and from time to time, a segment of output letters is produced, continuously building an infinite output word $\beta$ such that $(\alpha,\beta) \in R$.
We relate this to $\mathit{Win}$.
Recall that $R$ is defined by $\mathcal T$, and $\mathrm{dom}(R)$ by $\mathcal D$.
Given a play $\gamma$, there is a unique input word $\alpha\in\Sigma^\omega$ that corresponds to $\gamma$.
Since we are looking to build a computable function $f$ with $\mathrm{dom}(f) = \mathrm{dom}(R)$, we care whether $\alpha \in \mathrm{dom}(R)$.
The $\omega$-word $\mathit{col}_\mathcal D (\gamma)$ is equal to $c(\rho_\mathcal D)$, where $\rho_\mathcal D$ is the run of $\mathcal D$ on $\alpha$.
If $\mathrm{max}\ \mathrm{Inf}(\mathit{col}_\mathcal D (\gamma))$ is even, $\alpha \in L(\mathcal D)$, i.e., $\alpha \in \mathrm{dom}(R)$.
An output word $\beta \in \Gamma^\infty$ that corresponds to $\gamma$ is only indirectly defined, instead the play defines a (possibly finite) sequence of state transformations that an output word $\beta$ should induce together with $\alpha$.
How to extract a concrete $\beta$ from $\gamma$ is formally defined in the proof of \cref{thm:gametofunc}.
The $\omega$-word $\mathit{col}_\mathcal T (\gamma)$ contains the relevant information to determine whether $(\alpha,\beta) \in R(\mathcal T)$, i.e., $(\alpha,\beta) \in R$.
In particular, if $\beta$ is finite, $\mathrm{max}\ \mathrm{Inf}
\mathit{col}_\mathcal T (\gamma)$ is $-1$, that means that only
finitely many producing actions have been taken.
If $\mathrm{max}\ \mathrm{Inf} \mathit{col}_\mathcal T (\gamma)$ is even, we have that $(\alpha,\beta) \in R$.
Thus, $\mathit{Win}$ expresses that if $\alpha \in L(\mathcal D)$, then there is some $\beta \in \Gamma^\omega$, which can be built continuously while reading $\alpha$ such that $(\alpha,\beta) \in R$.

We make some remarks about the form of the game, in particular the use of two lookahead profiles, instead of one.
Assume we would have only one profile abstracting the lookahead over Adam inputs.
For simplicity, assume the specification is automatic (i.e., letter-to-letter).
Suppose, so far, Adam and Eve have alternated between providing an input letter and producing an output letter (in the finite-state game, Eve producing letter(s) corresponds to the abstract action of picking state transformations), but now, she needs to wait for more inputs before she can safely output something new.
Suppose that Adam has provided some more input, say the word $u$, and Eve now has enough information about the input to be able to produce something new.
Abstractly, it means that in the game, Adam has given the word $u$ but only its profile $P$ is stored.
Eve might not be able to produce an output of the same length as $u$ (for example, if producing the $i$-th output letter depends on the $(i+k)$-th input letter).
So, she cannot consume the whole profile $P$ (i.e., pick a state transformation in $P$).
What she has to do, is to decompose the profile $P$ into two profiles such that $P = P_1 \otimes P_2$ and pick a state transformation in $P_1$, and then continue the game with profile $P_2$ (and keep on updating it until she can again produce something).
The problem is, firstly, that there is no unique way of decomposing $P$ as $P_1 \otimes P_2$, and secondly, $P_1$ might not correspond to any prefix of $u$.
That is why it is needed to have explicitly the decomposition at hand in the game construction.

\subparagraph{From winning strategies to uniformizers.}
We are ready to state our first positive result: If Eve has a winning strategy in the unbounded delay game $\mathcal G_\mathcal T$, then $R(\mathcal T)$ is uniformizable by a computable function.
In fact, we show a more precise result, namely, that if Eve has a winning strategy, then the relation is uniformizable by a function recognized by a deterministic two-way parity transducer.
Additionally, if the domain of the relation is closed, then a deterministic one-way transducer suffices (see \cref{sec:closed}).
Just as (one-way) transducers extend parity automata with
outputs on their transitions, input-deterministic two-way transducers extend deterministic two-way parity automata with outputs. The reading tape is two-way, but the
output tape is one-way. The class of functions recognizable by 2DFTs
is smaller than the class of computable functions and enjoys many
good algorithmic properties, e.g., decidability of the equivalence
problem~\cite{alur2012regular}.
Note that any function recognizable by a 2DFT is computable, in the sense that it suffices to ``execute'' the 2DFT to get the output.
So, from now on, we may freely say that a function is computable by a 2DFT.

Before stating our first theorem, we remind the reader of \cref{remark:domain}.
As stated, it is not necessary (nor feasible in finite time) to check whether an input stream belongs to the domain.
However, since we employ the notion of uniformization to state our results, we need to properly define a domain of our uniformizer.
Hence, the 2DFT we construct (in the proof of \cref{thm:gametofunc}) has an acceptance condition corresponding to the desired domain.
In order to get the output from the 2DFT in a streaming fashion, the acceptance condition can simply be ignored.

\begin{theorem}
\label{thm:gametofunc}
  Let $R$ be defined by a transducer $\mathcal T$.
  If Eve has a winning strategy in $\mathcal G_\mathcal T$, then $R$ is uniformizable by a function computable by a 2DFT.
\end{theorem}

The proof comprises two parts.
First, we show that if Eve has a winning strategy in $\mathcal G_\mathcal T$, then $R(\mathcal T)$ is uniformizable by a computable function~$f$, see \cref{lemma:strat2func}.
Secondly, we show that the function~$f$ is also computable by a 2DFT, see \cref{lemma:strat22DFT}.

\subsection{Extracting a Computable Function from a Winning Strategy}

\begin{lemma}\label[lem]{lemma:strat2func}
  If Eve has a winning strategy in $\mathcal G_\mathcal T$, then $R(\mathcal T)$ is uniformizable by a computable function $f$.
\end{lemma}

\begin{proof}
Let $\sigma_\exists$ be a winning strategy for Eve in $\mathcal G_\mathcal T$.
This implies that there exists also a finite state winning strategy
for Eve, \modif{because $\mathcal G_\mathcal T$ is an $\omega$-regular game,
and $\omega$-regular games are finite-memory determined, by
B\"uchi-Landweber's Theorem}.

In \cref{alg:cont-function}, we give an algorithm that describes a computable function $f\colon \Sigma^\omega \to \Gamma^\omega$ that uniformizes $R$ which is constructed from a finite state winning strategy given by a strategy automaton $\mathcal S$.
Intuitively, the algorithm plays the game where Adam's actions correspond to an input word $\alpha$ and Eve's actions are determined by the winning strategy given by $S$.
In a play, the lookahead on the input is not stored concretely, but as an abstraction in form of two profiles.
When Eve takes an edge labeled with a state transformation, then the first lookahead profile is removed, the chosen state transformation is applied, the second lookahead profile takes the place of the first one, and the second lookahead profile is built anew.
In the algorithm, additionally, the lookahead on the input that is used to build the two profiles is stored concretely.
When Eve takes an edge labeled with a state transformation, say $(p,q,c)$, we have the information what kind of state transformation should occur for the input $u_1 \in \Sigma^*$ that currently corresponds to the first profile $P_1$.
The algorithm picks an output $v_1 \in \Gamma^*$ such that $\mathcal T\colon p \xrightarrow{u_1 / v_1} q$ and the maximal priority seen on this run is $c$.
Then, in the game, profile $P_2$ takes the place of $P_1$ and the second profile is set to $P_\varepsilon$.
Additionally, the algorithm sets $u_1$ to $u_2$, that is, the input that currently corresponds to $P_2$, and sets $u_2$ to $\varepsilon$.
The stored strings $u_1$ and $u_2$ can get arbitrarily long, but are reset from time to time, as the algorithm is built from a winning strategy, thus producing edges are taken again and again.
Note that the algorithm is guaranteed to produce non-empty output strings again and again, because in \cref{alg:append} a new letter is appended to the string $u_2$ again and again, and in \cref{alg:switch} the string $u_1$ is set to $u_2$ from time to time.
Thus, the string $u_1$ is infinitely often not empty.
In fact, $u_1$ is only empty the first time an output word is chosen.

The correctness of the algorithm can easily be seen.
Assume that $\gamma$ is a play according to the winning strategy $\sigma_\exists$ and assume that the corresponding input word $\alpha \in \Sigma^\omega$, spelt by Adam's actions, is part of the domain of the relation, i.e., $\alpha \in L(\mathcal D)$.
Hence, since $\mathrm{max}\ \mathrm{Inf}(\mathit{col}_\mathcal D (\gamma))$ is even, also $\mathrm{max}\ \mathrm{Inf}(\mathit{col}_\mathcal T (\gamma))$ is even, because $\gamma \in \mathit{Win}$.
This means that that Eve takes infinitely many producing edges, otherwise $\mathrm{max}\ \mathrm{Inf}(\mathit{col}_\mathcal T (\gamma))$ would be $-1$.
Thus, the algorithm infinitely often prints a finite output word building an output word $\beta$.
We have that $(\alpha,\beta) \in R$, i.e., $(\alpha,\beta) \in R(\mathcal T)$, because $\mathrm{max}\ \mathrm{Inf}(\mathit{col}_\mathcal T (\gamma))$ is even and $\beta$ was build such that there exists a run $\rho_\mathcal T$ of $\mathcal T$ on $\alpha / \beta$ such that $\mathrm{max}\ \mathrm{Inf}(\mathit{col}_\mathcal T (\gamma)) = \mathrm{max}\ \mathrm{Inf}(c(\rho_\mathcal T))$.
(Note that we assume that $\mathcal T$ defines a relation $R \subseteq \Sigma^\omega \times \Gamma^\omega$ (controlled with the parity condition) which guarantees that $\beta$ is build from finite words of which infinitely many are non-empty, because $\mathrm{max}\ \mathrm{Inf}(c(\rho_\mathcal T))$ is even.)

We now argue that algorithm yields a computable function.
The algorithm handles an input sequence in a letter-by-letter fashion,
the only ambiguity in the algorithm is in \cref{alg:chose-output}, 
that requires the algorithm to pick an output block $v_1$ such that $\underbrace{\mathcal T\colon p \xrightarrow{u_1 / v_1} q}_{\text{max prio is }c}$.
We explain that this is indeed computable by a deterministic machine as described in \cref{def:computable}.
Clearly, for each $(p,q,c)$ there exists a transducer, say $\mathcal T_{(p,q,c)}$, that recognizes $\{ (u,v) \mid \underbrace{\mathcal T\colon p \xrightarrow{u / v} q}_{\text{max prio is }c}\}$.
To pick some $v_1$ for a given $u_1$, a machine $M$ as in \cref{def:computable} can work as follows.
It is easily seen that the set $L = \{ v \mid \underbrace{\mathcal T\colon p \xrightarrow{u_1 / v} q}_{\text{max prio is }c}\}$ is regular, and hence computing some $v_1$ can be done by using a standard non-emptiness checking algorithm on $L$.
\end{proof}

\begin{algorithm}[t]
  \begin{algorithmic}
    \REQUIRE $\alpha\in\Sigma^\omega$, $G$ game arena, $\mathcal S = (M,V,m_0,\delta,\mu)$ strategy automaton

    \STATE $m \gets m_0$ \COMMENT{current state of the strategy automaton}
    \STATE $u_1 \gets \varepsilon$ \COMMENT{first input block}
    \STATE $u_2 \gets \varepsilon$ \COMMENT{second input block}
    \STATE $s_{\mathit{prev}} \gets s_0$, initial vertex of $G$ \COMMENT{previous vertex in the game}
    \STATE $s_{\mathit{cur}} \gets s_0$ \COMMENT{current vertex in the game}
    \STATE $a \gets \alpha(0)$ \COMMENT{current letter of $\alpha$}
    \STATE $action$ \COMMENT{current action of Eve}
    \WHILE{\TRUE}
    \STATE $u_2 \gets u_2.a$ \label[line]{alg:append} \COMMENT{append letter to second block}
    \STATE $s_{\mathit{cur}} \gets s$ if $s_{\mathit{cur}} \xrightarrow{a} s \in E$ \COMMENT{update game vertex according to Adam's action}
    \STATE $m \gets \delta(m,s_{\mathit{cur}})$ \COMMENT{strategy automaton is updated with Adam's action}
    \STATE $s_{\mathit{prev}} \gets s_{\mathit{cur}}$
    \STATE $action \gets \mu(m,s_{\mathit{cur}})$ \COMMENT{strategy automaton yields Eve's action}
    \STATE $s_{\mathit{cur}} \gets s$ if $s_{\mathit{cur}} \xrightarrow{action} s \in E$ \\\COMMENT{updated game vertex is of the form $(\cdot,\cdot,P_{u_1},P_{u_2},\cdot,\cdot)$}
    \STATE $m := \delta(m,s_{\mathit{cur}})$ \COMMENT{strategy automaton is updated according to Eve's action}
    \IF{$action$ is of the form $(p,q,c)$ \COMMENT{Eve took a producing edge}} 
    \STATE choose output block $v_1 \in \Gamma^*$ such that $\mathcal T\colon p \xrightarrow{u_1 / v_1} q$ with max prio $c$ \label[line]{alg:chose-output}
    \STATE   $u_1 \gets u_2$ \label[line]{alg:switch} \COMMENT{first input block becomes second}
    \STATE $u_2 \gets \varepsilon$  \COMMENT{second block is emptied}
    \PRINT $v_1$ \COMMENT{produce output block} \label[line]{alg:print}
    \ENDIF
    \STATE $a \gets \alpha.\mathit{nextLetter()}$ \COMMENT{read next input letter}
    \ENDWHILE

    \ENSURE $\beta \in \Gamma^\infty$, if $\alpha \in \mathrm{dom}(R)$, then $(\alpha,\beta) \in R$
    \end{algorithmic}
    \caption{Algorithm computing a function that uniformizes $R$. The algorithm is described in the proof of \cref{lemma:strat2func}.}
    \label{alg:cont-function}
 \end{algorithm}

 \subsection{Extracting a 2DFT from a Winning Strategy}


We have seen that \cref{alg:cont-function} uses unbounded memory.
The following lemma shows that this unbounded memory can be traded to for finite memory at the cost of having input two-wayness.

\begin{lemma}\label[lem]{lemma:strat22DFT}
    If Eve has a winning strategy in $\mathcal G_\mathcal T$, then $R(\mathcal T)$ is
    uniformizable by a~2DFT.
\end{lemma} 

  \begin{proof}
    We first give an intuition how to translate a winning strategy into a 2DFT.
    The main idea is to use two-wayness to encode finite lookahead over the input: the reading head goes forward to gather input information, and then must return to the initial place where the lookahead was needed to transform the input letter.
    The difficulty is for the 2DFT to return to the correct position, even though the lookahead can be arbitrarily long.
    
    In order to find the correct positions, we make use of a finite-state strategy automaton for Eve's winning strategy in the following sense.
    A (left-to-right) run of the strategy automaton on the input word yields a unique segmentation of the input, such that segments $i$ and $i+1$ contain enough information to determine the output for segment $i$.
    The idea is to construct a 2DFT that simulates the strategy automaton in order to find the borders of the segments.
    If the 2DFT goes right, simulating a computation step of the deterministic strategy automaton is easy.
    Recovering the previous step of the strategy automaton when the 2DFT goes left is non-trivial, it is possible to compute this information using the Hopcroft-Ullman construction presented in~\cite{hopcroft1967approach}.
    We show that having the knowledge of the profiles of segments $i$ and $i+1$ is enough to deterministically produce a matching output for segment $i$ on-the-fly going from left-to-right over segment $i$ again.

    Now, we present a detailed proof.
      Assume that Eve has a finite-state winning strategy in $\mathcal G_\mathcal T$.
      \cref{alg:cont-function} describes a computable function
      that uniformizes $R(\mathcal T)$ built from a finite-state winning strategy
      given by strategy-automaton $\mathcal S$. At \cref{alg:chose-output}, this algorithm
      makes a choice of output $v_1$. We slightly modify this choice: given
two states $p$ and $q$ of the transducer recognizing $R(\mathcal T)$ and a
priority $c$, the relation
\[
    R_{(p,q,c)} = \{ (u,v) \mid \mathcal T\colon p \xrightarrow{u_/v} q \text{ and the maximal priority seen is $c$}\}.
\]
is a rational relation of finite words. It is known that every
rational relation on finite words has an effective rational uniformizer, see, e.g., \cite{DBLP:books/lib/Berstel79}. So let
$f_{(p,q,c)}$ denote such a rational function, as obtained
in~\cite{DBLP:books/lib/Berstel79}. We modify \cref{alg:chose-output} of the algorithm
so that it chooses output $v_1 = f_{(p,q,c)}(u_1)$.

We denote by $f$
the function computed by the (modified) algorithm. We show that $f$ is recognizable by an
input-deterministic two-way transducer.
The main difficulty is that in the game, there is a desynchronization
between the moment Eve makes a producing action (a triple $(q,q',c')$) and the interval of the play during which the
input corresponding to this producing action has been provided by Adam. In
\cref{alg:cont-function}, this is easily solved because
unbounded memory can be used to store these intervals. An input-deterministic two-way transducer must solve this by reading far enough into the input such that the producing action is known, then the transducer must deterministically go back to the input segment for which this producing action was chosen. 
The main difficulty lies in deterministically finding the right beginning and end points of such a segment when going back.

We provide a construction in two steps.
First, we describe a simpler setting, where an input sequence is enriched by a function $f_{annot}$ with information about the corresponding play in the game consistent with the strategy defined by the strategy automaton.
Those information are: vertices of the play, actions
performed during the play, current state of the strategy automaton. Recall that the game alternates
between Adam and Eve actions and start by
an Adam action. Adam's actions are input symbols in $\Sigma$ while Eve's
actions are either $skip$ or a triple $(q,q',c')$ of some profile. Let
us denote by $A_\exists$ Eve's actions. Then, the type of $f_{annot}$ is
\[
  f_{annot}\ :\ \text{dom}(R)\subseteq \Sigma^\omega\rightarrow
\Bigl(\Sigma \times \bigl((V \times M)\times (V \times M)\times (A_\exists \times V \times M)\bigr)\Bigr)^\omega
\]
where $V$ are the vertices of the game and $M$ are the states of the strategy automaton.
The function returns an input sequence annotated with sufficient information to reconstruct a play in the game
induced by the strategy applied on this input sequence, and the run of the strategy automaton on this
play. More
precisely, for a single letter $a$ the annotation contains the vertex of the
game and the state of the strategy automaton, Adam's action (which is
$a$), the target vertex and target state induced by Adam's action,
Eve's action as prescribed by the strategy automaton, the target
vertex and target state induced by Eve's action.

We explain how to construct an input-deterministic two-way transducer that
recognizes the function $g$ which takes an annotated input sequence
$f_{annot}(\alpha)$ and maps it to $f(\alpha)$ for all $\alpha \in \mathrm{dom}(R)$.
Subsequently, we explain how to modify the transducer to recognize the function $f\colon \alpha \mapsto f(\alpha)$ for all $\alpha \in \mathrm{dom}(R)$.

An annotated string $\alpha'$ based on $\alpha$ defines a unique segmentation $u_0u_1\cdots$ of $\alpha$ where the endpoint of $u_i$ is annotated with a producing action of Eve for all $i \geq 1$ ($u_0$ is defined as $\varepsilon$).
The producing action $(q,q',c')$ seen at the end of $u_i$ means that for the segment $u_{i-1}$ an output $v_{i-1}$ must be produced such that $\mathcal T\colon q \xrightarrow{u_{i-1}/v_{i-1}} q'$ and the maximal priority seen is $c'$.
In the annotated sequences it is easy for an input-deterministic two-way transducer to advance to the end of the segment $u_i$ and go back to the beginning of segment $u_{i-1}$ for all $i\geq 1$ because all segment endpoints are annotated with a producing action (the intermediate positions are annotated with $skip$).

We now explain how an input-deterministic two-way transducer can produce the
desired output $v_{i-1}$ for the segment $u_{i-1}$ with
annotations. First, we show that for any state transformation
$(q,q',c')$, the function $f_{(q,q',c')}$ of finite words can be
computed by an input-deterministic two-way transducer.

Rational functions over finite words are recognizable by input-deterministic two-way transducers, which can be seen as follows.
In \cite{DBLP:conf/icla/Filiot15}, it is shown that a function from finite words to finite words is rational iff it is definable as an order-preserving MSO-transduction.
In \cite{engelfriet2001mso}, it is shown that a function from finite words to finite words is definable as an MSO-transduction iff it is definable by an input-deterministic two-way transducer.
Hence, it follows that rational functions over finite words are recognizable by input-deterministic two-way transducers.

Such a transducer can be constructed. Let $\mathcal T_{(q,q',c')}$ denote the input-deterministic two-way transducer that recognizes $f_{(q,q',c')}$.

Consequently, when the input-deterministic two-way transducer on $\alpha'$
has returned to the beginning of the annotated segment $u_{i-1}$
(after it went to the endpoint of $u_i$ to get the state
transformation $(q,q',c')$ corresponding to $u_{i-1}$), it can run the input-deterministic two-way transducer $\mathcal T_{(q,q',c')}$ on $u_{i-1}$ that will produce a suitable $v_{i-1}$.
The difficulty is to stay inside of $u_{i-1}$ for this transformation.
However, the annotations in $\alpha'$ clearly mark the endpoints of
$u_{i-2}$ and $u_{i-1}$ (by state transformations instead of $skip$
action) , so we can run the transducer $\mathcal T_{(q,q',c')}$
designed to work on finite words because we have clear markers on the
bounds of the segment~$u_{i-1}$.

After producing the output $v_{i-1}$, the two-way transducer can advance to the end of $u_{i+1}$, obtain the producing action that should be applied to $u_{i}$, go back to the beginning of $u_i$, apply it and so on.

The priorities that have to be seen on a run of this transducer can be derived from the annotations.
Annotation $skip$ means priority $-1$, and an annotation of the form $(q,q',c')$ means priority $c'$ must be seen.
Thus, we have shown that there exists an input-deterministic two-way
transducer that recognizes $g$. Let $\mathcal T_g$ denote this transducer.

We now explain how we can get rid of the annotations.
When we read an input $\alpha\in \Sigma^\omega$, we can determine the
annotations assigned by $f_{annot}$ deterministically on-the-fly using
the strategy automaton, because the strategy automaton is deterministic.
Therefore we slightly have to extend the strategy automaton such that
it additionally stores in its state space the current vertex (that is
the input alphabet of the strategy automaton) and Eve's chosen action (determined by the strategy automaton via it's next move function).

Let $\mathcal A$ denote this enhanced strategy automaton. It remains
deterministic. The deterministic automaton $\mathcal A$ has the
property that when going from left-to-right in $\alpha$, it computes
the annotation in the sense that it corresponds to the information
stored in the state space of $\mathcal A$. So, by running
$\mathcal{A}$ from left-to-right, we have access to the annotation.

Our goal is now to extend $\mathcal T_g$ such that it computes the run of $\mathcal A$ on $\alpha$.
Clearly, when going from a position $j$ to $j+1$ is is possible to compute the state of $\mathcal A$ on $j+1$ when the state of $\mathcal A$ on $j$ is known, however, when going from a position $j$ to $j-1$, we also want to obtain the state of $\mathcal A$ that was seen on $j-1$ based on the state that was seen on $j$.

It is possible to compute this information with an input-deterministic two-way transducer using the Hopcroft-Ullman construction presented in \cite{hopcroft1967approach}.
The construction describes how to simulate a deterministic one-way automaton together with a two-way automaton.
It is described for two-way automata on finite words, however, in
order to obtain the state of the one-way automaton at position $j-1$
when one is at position $j$, the Hopcroft-Ullman construction only
needs to re-read some positions to the left of $j$, so the same
construction can be applied in the infinite word setting.

Hence, $\mathcal T_g$ combined with $\mathcal A$ using the Hopcroft-Ullman construction can be run on words $\alpha\in\Sigma^\omega$ instead of $f_{annot}(\alpha)$ and computes the desired function $f$.
\end{proof}

We like to mention that \cref{lemma:strat22DFT} could have been obtained in another way.
Carton and Dou{\'{e}}neau{-}Tabot \cite{DBLP:conf/mfcs/CartonD22} have shown that every computable rational function can be computed by some 2DFT.
Their result provides a translation from computable rational functions to 2DFTs going through intermediate computation models.
To employ their result, we need to show that the computable function $f$ obtained from a winning strategy for Eve is rational (i.e., a non-deterministic one-way transducer recognizing $f$ exists).
The proof is similar to the proof of \cref{lemma:strat22DFT}.
The main difference is that we use the two-wayness to deterministically determine a priori (by obtaining some lookahead on the input) which input transformation should be realized on a finite input segment.
Without two-wayness, we use non-determinism to guess and output the correct input transformation.
The correctness is then verified a posteriori.

However, we decided in order to obtain \cref{lemma:strat22DFT} to give
a direct construction to be self-contained and build a 2DFT with
better \modif{state} complexity than the 2DFT that would result from the translations given in \cite{DBLP:conf/mfcs/CartonD22}.

\subsection{Complexity Analysis}\label{app:complexity}

\begin{lemma}
\label[lem]{lemma:complexity}
  Deciding whether Eve has a winning strategy in $\mathcal G_\mathcal T$ is in \upshape{\textsc{ExpTime}}.
\end{lemma}

\begin{proof}
Two-player $\omega$-regular games are decidable (see, e.g., \cite{gradel2002automata}).
The claimed upper bound is achieved by representing the winning condition as a deterministic parity automaton, carefully analyzing its size, and then solving a parity game.

From $\mathcal{G}_\mathcal T$ we can obtain an equivalent parity game of size $|G| \cdot |\mathcal W|$, where $\mathcal W$ is a deterministic parity automaton that recognizes the set $\mathit{Win}$.

First, we analyze the size of $G = (V,E)$.
A vertex is an element of 
\[
  Q_\mathcal T \times (C_\mathcal T \cup \{-1\}) \times \mathcal P \times \mathcal P \times Q_\mathcal D \times \{\adam,\eve\},
\]
where $\mathcal P$ is the set of all profiles.
A profile is a subset of $Q_\mathcal T \times Q_\mathcal T \times C_\mathcal T$, thus, the set of profiles is of size exponential in $|Q_\mathcal T|$ and $|C_\mathcal T|$.
We assume that the domain automaton $\mathcal D$ is obtained from $\mathcal T$ by projection and determinization.
A deterministic parity automaton with $2^{\mathcal O (n \ \mathrm{log}\ n)}$ states and $\mathcal O(nk)$ priorities can be constructed from a nondeterministic parity automaton with $n$ states and $k$ priorities using Safra's construction \cite{safra1992exponential}.
We obtain that $\mathcal D$ has a state space of size exponential in $|Q_\mathcal T|$ and priorities of size linear in $|Q_\mathcal T| \cdot |C_\mathcal T|$.
It follows that $V$ and $E$ are of size exponential in $|Q_\mathcal T|$ and $|C_\mathcal T|$.
All in all, since $|C_\mathcal T| \leq |Q_\mathcal T|$, $|V|$ and $|E|$ are exponential in $|Q_\mathcal T|$.

Secondly, we analyze the size of $\mathcal W$.
Recall that $\mathit{Win}$ contains plays $\gamma$ that satisfy
\[
    \mathrm{max}\ \mathrm{Inf}(\mathit{col}_\mathcal D (\gamma)) \text{ is even} \rightarrow \mathrm{max}\ \mathrm{Inf}(\mathit{col}_\mathcal T (\gamma)) \text{ is even}.
\]
We reformulate this as
\[
    \mathrm{max}\ \mathrm{Inf}(\mathit{col}_\mathcal D (\gamma)) \text{ is odd} \ \vee\  \mathrm{max}\ \mathrm{Inf}(\mathit{col}_\mathcal T (\gamma)) \text{ is even}.
\]
A deterministic parity automaton for the second disjunct has a state set of size linear in $|Q_\mathcal T|$ with priorities $C_\mathcal T$.
The first disjunct simply states that the input sequence $\alpha$ associated to $\gamma$ is in $\mathrm{dom}(f)$, i.e., $\alpha \in L(\mathcal D)$.
A nondeterministic parity automaton for $L(\mathcal D)$ has a state set of size linear in $|Q_\mathcal T|$ and priorities $C_\mathcal T$.
Thus, a nondeterministic parity automaton for the conjunction has a state set of size linear in $|Q_\mathcal T|$ and priorities $C_\mathcal T$.
Using Safra's construction, we obtain a deterministic variant with a state set of size exponential in $|Q_\mathcal T|$ with priorities of size linear in $|Q_\mathcal T| \cdot |C_\mathcal T|$.
Let $\mathcal W$ denote this automaton with state set $Q_\mathcal W$.
The alphabet of $\mathcal W$ is $V$, thus, the alphabet size is exponential in $|Q_\mathcal T|$.

The parity game with arena $G \times \mathcal W$ has $|V|\cdot|Q_\mathcal W|$ many vertices and $|E|\cdot|Q_\mathcal W|$ many edges.
Thus, its vertices and edges are of size exponential in $|Q_\mathcal T|$.
The game has priorities of size linear in $|Q_\mathcal T| \cdot |C_\mathcal T|$.

A parity game with $n$ vertices, $m$ edges, and $k$ priorities can be solved in time $\mathcal O(mn^{k/3})$, see \cite{schewe2007solving}.
Consequently, we obtain that our parity game is solvable in time exponential in $|Q_\mathcal T| \cdot |C_\mathcal T|$.
Since $|C_\mathcal T| \leq |Q_\mathcal T|$, the overall time complexity is exponential in $|Q_\mathcal T|$.
\end{proof}
\subsection{Extracting a 1DFT from a Winning Strategy with Bounded Delay}

We state a lemma about bounded delay. First, we introduce the notion of bounded  delay 1DFT.
Let $\mathcal T = (Q,\Sigma,\Gamma,q_0,\Delta,c)$ be a 1DFT and $K\in\mathbbm{N}$. We say that~$\mathcal T$ is \emph{$K$-delay} if, intuitively, the output production is late at most $K$ steps and never ahead, i.e., for all runs $(p_0,u_0,v_0,p_1)(p_1,u_1,v_1,p_2)\dots (p_{n-1},u_{n-1},v_{n-1},p_n)\in \Delta^*$, it holds that $|u_0u_1\dots u_{n-1}|-|v_0v_1\dots v_{n-1}|\leq K$.

\begin{lemma}
  \label[lem]{lemma:boundedstrat}
  Let $R$ be defined by a transducer $\mathcal T$.
  If there exists $\ell \geq 0$, such that Eve has a winning strategy
  in $\mathcal G_\mathcal T$ with at most $\ell$ consecutive
  \emph{skip}-moves, then $R$ is uniformizable by a function $f$
  computable by a 1DFT $\mathcal T_f$ . Moreover, if $\mathcal T$ is
  letter-to-letter, then $\mathcal T_f$ is $\ell$-delay.
\end{lemma}

\begin{proof}
Intuitively, such a strategy yields a function computable by a 1DFT, because the needed lookahead (as it is bounded) can be stored in the state space.

    Let $\sigma_\exists$ be such a winning strategy for Eve in $\mathcal G_\mathcal T$ (i.e., Eve never makes more than $\ell$ consecutive \emph{skip}-moves).
  We can use \cref{alg:cont-function} to obtain a computable function $f$ that uniformizes $R(\mathcal T)$.
  We now explain that $f$ is in fact sequential, i.e., recognizable by a 1DFT.
  
  The construction of \cref{alg:cont-function} yields that $f$ is such that it reads a non-empty input block of length at most $\ell$, and then produces an output block (a canonical choice of size linear in $\ell$ is possible), and so on.
  It directly follows that $f$ is sequential, and a 1DFT $\mathcal
  T_f$ that
  recognizes $f$ can be effectively constructed. 
  If additionally $\mathcal T$ is letter-to-letter, then the
  strategy, after
  reading a non-empty input block of length at most $\ell$, picks an output block of the same size.
  Therefore, $\mathcal T_f$ is $\ell$-delay. 
  \phantom{blub}
  \end{proof}
  
  \subsection{A Note on the Incompleteness of the Decision Procedure}

  \begin{remark}
    The converse of \cref{thm:gametofunc} is not true.
  \end{remark}

  \modif{First, the existence of a computable uniformizer does not imply the
  existence of a winning strategy, otherwise, the $\RAT$-synthesis
  problem would be decidable, which is a contradiction to
  \cref{thm:undecidable}. The following example strengthens this
  result: it shows that even if there is a 2DFT uniformizer, it does
  not necessarily imply the existence of a winning strategy.}

  \begin{example}
    Consider the identity function $f\colon \{a,b\}^\omega \to \{a,b\}^\omega$ such that all inputs with either finitely many $a$ or $b$ are in the domain.
    A (badly designed) letter-to-letter transducer $\mathcal T$ that recognizes $f$ has five states $S,A,B,A',B'$, where $S$ is the starting state, $A,B$ (resp.\ $A',B'$) are used to recognize finitely many $b$ (resp.\ $a$), and from $S$, the first input/output letter non-deterministically either enters $A$ or $A'$.
    In a play in $\mathcal G_\mathcal T$, at some point, Eve must make her first output choice, i.e., she starts to build a run of $\mathcal T$.
    This choice fixes whether the run is restricted to $A,B$ or $A',B'$.
    No matter Eve's choice, Adam can respond with an infinite sequence of either only $a$ (for $A,B$) or $b$ (for $A',B'$), making it impossible to build an accepting run.
    Thus, Eve has no winning strategy, but clearly the function $f$ is
    computable \modif{by a 2DFT}.
  \end{example}
  
  While in the above example, the point of failure is clearly the bad presentation of the specification, this is not the case in general.
  Recall the proof sketch of \cref{thm:undecidable}, where we provide a reduction from Post's correspondence problem.
  A non-deterministic transducer constructed from a given PCP instance $u_1,v_1,\dots,u_n,v_n$ can guess whether the input word contains infinitely many $a$, and accordingly either checks that the output begins with $u_{i_1}\cdots u_{i_k}$ for input sequences beginning with indices $i_1\cdots i_k$, or checks that it does not begin with a prefix equal to $v_{i_1}\cdots v_{i_k}$.
  As detailed in the proof sketch, if the PCP instance has no solution, there is a computable uniformization, however, using the same argumentation as in the above example, such a transducer would make it impossible to have a winning strategy.
  In order to have a winning strategy, the transducer must be changed such that it checks at the same time whether the output starts with $u_{i_1}\cdots u_{i_k}$ and does not start with something equal to $v_{i_1}\cdots v_{i_k}$ for input sequences beginning with $i_1\cdots i_k$. 
  In general, depending on the PCP instance, it is not possible to make both checks in parallel.
\section{A Sufficient Condition for Completeness}
\label{sec:completeness}

\modif{\cref{thm:gametofunc} yields a procedure which is sound but not
complete for solving the \RAT-synthesis problem. From a transducer
$\mathcal T$, construct the $\omega$-regular game
$\mathcal{G}_{\mathcal{T}}$ and solve it: if Eve wins, then $\mathcal
T$ is uniformizable by a computable function (even computable by a
2DFT). Otherwise, nothing can be concluded.}

In this section, we show that the procedure is complete for two known and expressive classes of rational relations, namely the class of automatic relations ($\AUT$), which are for example used as specifications in Church synthesis, as well as the class of deterministic rational relations ($\DRAT$) \cite{sakarovitch2009elements} (to be formally defined below), see \cref{cor:main}.

To arrive at these results, we define a structural restriction on transducers that turns out to be a sufficient condition for completeness.
Let $\mathcal T$ be a transducer.
An \emph{input} (resp.\ \emph{output}) state is a state $p$ from which there exists an outgoing transition $(p,u,v,q)$ such that $u\neq\varepsilon$ (resp.\ $v\neq\varepsilon$).
The set of input (resp.\ output) states is called $Q_i$ (resp.\ $Q_o$).
A transducer $\mathcal T$ has \emph{property $\mathcal P$} if 
\modif{
the following two conditions hold:
\begin{enumerate}
  \item The transition set is a finite subset of $Q \times (\Sigma \cup \{\varepsilon\}) \times \Gamma^* \times Q$.
  In words, a transition reads one or no input letters.\footnote{We thank the anonymous reviewer for alerting us about the necessity of this conditon.}
  \item For all words $u\in\Sigma^*$, all $v_1,v_2\in\Gamma^*$ such that $v_1$ is a prefix of $v_2$ the following holds:
  \begin{equation*}
     \text{ if } \mathcal T\colon p \xrightarrow{u/v_1} q, \mathcal T\colon p \xrightarrow{u/v_2} r, \text{ and } q,r \text{ are input states, then } q = r.
  \end{equation*}
\end{enumerate}
}

\modif{
The property $\mathcal P$ ensures that given $\alpha \in
\Sigma^\omega$ and $\beta \in \Gamma^\omega$ such that there is a run
of $\mathcal T$ with input $\alpha$ and output $\beta$, for each
prefix $u \in \Sigma^*$ of $\alpha$ all runs $\mathcal T\colon q_0
\xrightarrow{u/v}$ that end in a input state such that $v \in
\Gamma^*$ is a prefix of $\beta$ lead to the same input state.
}
This implies that input prefixes together with (long enough) output
prefixes are sufficient to determine the beginning of an accepting run
up to the last state reached before reading the next input letter
following the prefix. This allows us to show the following.

\begin{theorem}
\label{thm:functogame}
Let $R$ be defined by a transducer $\mathcal T$ with property $\mathcal P$.
If $R$ is uniformizable by a computable function, then Eve has a winning strategy in $\mathcal G_\mathcal T$.
\end{theorem}

We give the formal proof of this result in \cref{app:functogame} and provide a proof sketch first to not interrupt the flow of this section.

\begin{proof}[Proof sketch]
  In fact, we explain how to construct a winning
  strategy from a continuous (a computable function is always continuous, see \cref{app:functogame}) uniformizer $f$ of $R$. Given $\alpha \in
  \mathrm{dom}(R)$ and $f(\alpha)$, we show that it is possible to decompose the input $\alpha$
  into $u_0u_1\cdots$ and the output $f(\alpha)$ into $v_0v_1\cdots$
  such that there exists an accepting run $\mathcal T\colon q_0 \xrightarrow{u_0/v_0} q_1
  \xrightarrow{u_1/v_1} q_2 \cdots$ where each $q_i$ is an input state for $i > 0$. Moreover, this decomposition and
  run can be determined in a unique way and on-the-fly, in the sense
  that a factor $u_i/v_i$ only depends on the factors
  $u_0/v_0$,$\ldots$, $u_{i-1}/v_{i-1}$. 
  This makes it possible for Eve to pick a corresponding state
  transformation sequence $(q_0,q_1,c_0)(q_1,q_2,c_1)\cdots$ which is
  only dependent on \emph{the so far seen} actions of Adam spelling
  $u_0u_1\cdots$. The main idea to determine the $u_i$ is to look at
  the indices $j$ for which the longest common prefix of the sets $S_j = \{
  f(\alpha({:}j)\beta)\mid \alpha({:}j)\beta\in\mathrm{dom}(R)\}$ strictly
  increases. Given $u_i$, the output $v_0 v_1 \cdots v_i$ is any common prefix of the sets $S_j$, such that a run $\mathcal T\colon q_i \xrightarrow{u_i/v_i}$ is defined and its target is an input state.
  The fact that $\mathcal T$ has property $\mathcal P$ guarantees that each of these runs has the same target, thus, the next state transformation $(q_i,q_{i+1},c_i)$ is uniquely determined.
\end{proof}

We formally introduce \AUT and \DRAT.
A relation is \emph{deterministic rational} if it is recognized by a transducer where $Q_{i}$ and $Q_o$ partition its state space, and its transition relation $\Delta$ is a function $\left(Q_i \times \Sigma \times \{\varepsilon\} \to  Q\right) \cup \left(Q_o \times \{\varepsilon\} \times \Gamma \to Q\right)$.
It is \emph{automatic} if additionally $\Delta$ strictly alternates between $Q_i$ and $Q_o$ states.
It is easy to see that every \DRAT-transducer (and a fortiori every \AUT-transducer) satisfies the property $\mathcal P$.
In general, given any transducer $\mathcal T$, we do not know if it is decidable whether $\mathcal T$ has property $\mathcal P$.

\subparagraph{Main result.} We now state our main result: Asking for the existence of a uniformization which is
computable by a Turing machine or computable by an input-deterministic two-way transducer (2DFT), are equivalent questions, as long as specifications are \DRAT relations.
Moreover, these questions are decidable. 

\begin{corollary}
\label[cor]{cor:main}
  Let $R$ be defined by a \DRAT-transducer $\mathcal T$.
  The following are equivalent:
  \begin{enumerate}
    \item $R$ is uniformizable by a computable function.
    \item $R$ is uniformizable by a function computable by a 2DFT.
    \item Eve has a winning strategy in $\mathcal G_\mathcal T$.
  \end{enumerate}
\end{corollary}

Note that the above result also holds for the slightly more general case of relations given by transducers with property $\mathcal P$.

\begin{theorem}
  \label{thm:complexity}
 The \AUT- and \DRAT-synthesis problems are \upshape{\textsc{ExpTime}}-complete.
  \end{theorem}

\begin{proof}
Membership in \textsc{ExpTime} directly follows from \cref{lemma:complexity,cor:main}.
In \cite{DBLP:journals/corr/KleinZ14} it was shown that this problem is \textsc{ExpTime}-hard in the particular case of automatic relations with total domain, so the lower bound
applies to our setting.
\end{proof}

\subsection{From Continuous Functions to Strategies}
\label{app:functogame}
In the upcoming proofs, we make use of the notion of continuity instead of computability, as this simplifies the proofs.
We first recall its definition.
A function is called \emph{continuous} if 
\begin{equation}\label{eq:continuous}
\forall \alpha \in \mathrm{dom}(f)\ \forall i \in \mathbbm N\ \exists j \in \mathbbm N\  \forall \beta \in \mathrm{dom}(f)\colon |\alpha \wedge \beta| \geq j \rightarrow |f(\alpha) \wedge f(\beta)| \geq i.
\end{equation}

It is easy to see that every computable function is also continuous.

\begin{example}
    Consider the function $f_1$ of \cref{ex:cont-uni}. 
    $f_1$ is continuous, because the $i$-th output
    symbol only depends on the $max(i,n+1)$ first input symbols. 
Consider the function $f_2$ of \cref{ex:cont-uni}. 
  The function $f_2$ is clearly rational, but it is not continuous.
  We verify that $f_2$ is not continuous, let $\alpha_n$ denote $a^nb^\omega$, we have that $|\alpha_n \wedge a^\omega| = n$ and $|f_2(\alpha_n) \wedge f_2(a^\omega)| = 0$ for all $n\in\mathbbm N$.
  Thus, $f_2$ is not continuous.
\end{example}

We introduce an auxiliary function used in the proof of \cref{thm:functogame} that yields the ``best'' (meaning most beneficial for Eve) priority associated to a profile and two states.

\begin{definition}[best]\label[def]{def:best}
We define the function $\mathit{best}\colon \mathcal P \times Q_\mathcal T \times Q_\mathcal T \to c_\mathcal T$, where $\mathcal P$ denotes the set of all profiles.
We let 
\[
    \mathit{best}(P,p,q) =  
    \begin{cases}
        \mathrm{max}\{ c \mid (p,q,c) \in P, \text{ $c$ is even}\} & \\
        & \hspace*{-3cm}\text{if there is an even $c'$ s.t.\ $(p,q,c') \in P$} \\
        \mathrm{min}\{ c \mid (p,q,c) \in P, \text{ $c$ is odd}\} & \\
        & \hspace*{-3cm}\text{if there is no even $c'$ s.t.\ $(p,q,c') \in P$}
    \end{cases}
\]
\end{definition}

The meaning of $\mathit{best}(P,p,q) = c$ is that it for an input word
$u$ that has profile~$P$, \modif{the best priority Eve can achieve
  when going from $p$ to $q$ via some $u/v$ is $c$. This $c$ is the
  maximal even priority if there is some, while $c$ is the minimal odd
  priority if there is no even priority that she can achieve}.

As already mentioned in the proof sketch of \cref{thm:functogame}, we show a stronger result, namely, the following result.

\begin{theorem}
  \label{lemma:func2strat}
  Let $R$ be defined by a transducer $\mathcal T$ with property $\mathcal P$.
  If $R$ is uniformizable by a continuous function $f$, then Eve has a winning strategy in $\mathcal G_\mathcal T$.
\end{theorem}

\begin{proof}
To begin with, we describe how to inductively factorize an input word $\alpha \in \Sigma^\omega$.
Let $u_0 = \varepsilon$, and let $u_0\cdots u_{i-1}$ be the already defined factors of $\alpha$, let $\alpha = u_0\cdots u_{i-1}\alpha'$, and let
\begin{equation}\label{eq:factor}
u_i := \alpha'(:\! k),
\end{equation}
where $k\in\mathbbm N_{\geq 1}$ is chosen such that $\alpha'(:\! k)$
is the shortest non-empty prefix of $\alpha'$ such that there is a prefix of $\hat{f}(u_0\cdots u_{i-1}\alpha'(:\! k))$, say $v$, that satisfies
$\Delta_\mathcal T^*(q_0^\mathcal T,u_0\cdots u_{i-1}/v) \cap
Q_i^\mathcal T \neq \emptyset$ (let $Q_i^\mathcal T$ denote the set of
input states of the transducer $\mathcal T$ \modif{as defined at the
beginning of \cref{sec:completeness}}).
In words, the input $u_0\cdots u_i$ provides enough information to determine a run of $\mathcal T$ that consumes the input $u_0\cdots u_{i-1}$.
If such a prefix of $\alpha'$ does not exist, let $u_i = \alpha'$, that is, the infinite remainder of $\alpha$.

Let $\alpha_1,\alpha_2 \in \Sigma^\omega$, and let $u_0^1\cdots$ and $u_0^2\cdots$ denote their factorizations built according to \cref{eq:factor}.
Note that $u_i^1 = u_i^2$ for all $i$ such that $|u_0^1\cdots u_i^1| \leq |\alpha_1 \wedge \alpha_2|$.
Furthermore, given an input word $\alpha$, the factorization according to \cref{eq:factor} can be determined on-the-fly reading the input word letter-by-letter.
If $\alpha \in \mathrm{dom}(f)$, its factorization has infinitely many
factors, because $f$ is continuous meaning that $\hat{f}$ grows over
time (\modif{$\hat{f}$ has been defined in \cref{sec:prelims}}).

In the remainder of the proof we are only interested in plays where the input word $\alpha$ spelt by Adam is in $L(\mathcal D)$, i.e., $\mathrm{dom}(R)$, because all other plays automatically satisfy $\mathit{Win}$, so there is nothing to show for plays whose corresponding input word $\alpha \notin L(\mathcal D)$.
Since $f$ uniformizes $R$, we have that $\mathrm{dom}(f) = \mathrm{dom}(R)$, so we can safely consider only plays where its corresponding input word $\alpha \in \mathrm{dom}(f)$.

The idea of the proof is that we use the sequence of input blocks $u_0u_1\cdots$ of $\alpha$ built as in \cref{eq:factor} to obtain the actions of Eve.
Her actions pick state transformations, our goal is to pick state transformations by defining a sequence of output blocks $v_0v_1\cdots$ and pick state transformations that are induced by $u_0 / v_0 \cdot u_1 / v_1\cdots$.
Formally, let $v$ be the shortest prefix of $\hat{f}(u_0\cdots u_i)$ such that $\Delta_\mathcal T^*(q_0^\mathcal T,u_0\cdots u_{i-1}/v) \cap
Q_i^\mathcal T \neq \emptyset$ and define
\begin{equation}\label{eq:v}
{v_{i-1}} = (v_0\cdots v_{i-2})^{-1}v.
\end{equation}

We define a strategy that satisfies the following invariant for all $i\in\mathbbm N$:

\begin{lemma}\label[lem]{claim:invariant}
If there exists a run $\rho$ of $\mathcal T$ on ${u_0}\cdots{u_{i-1}} / {v_0}\cdots{v_{i-1}}$ of the form
\[
\mathcal T\colon \underbrace{q_0^\mathcal T \xrightarrow{{u_0} / {v_0}} q_1 \in Q_i^\mathcal T}_{\rho_0} \quad \underbrace{q_1 \xrightarrow{{u_1} / {v_1}} q_2\in Q_i^\mathcal T}_{\rho_1} \quad \cdots \quad \underbrace{q_{i-1} \xrightarrow{{u_{i-1}} / {v_{i-1}}} q_i\in Q_i^\mathcal T}_{\rho_{i-1}},
\]
then the sequence of chosen producing edges in the play after the input $u_0\cdots u_i$ has the label sequence 
\[
(q_0^\mathcal T,q_1,c_0)(q_1,q_2,c_1)\cdots(q_{i-1},q_{i},c_{i-1}),
\]
where the edge labeled with $(q_j,q_{j+1},c_{j}) \in P_{u_j}$ with $c_j = \mathit{best}(P_{u_j},q_j,q_{j+1})$ (see \cref{def:best}) is taken after the input $u_0\cdots u_{j+1}$ for all $0 \leq j \leq {i-1}$.
Furthermore, after the input $u_0\cdots u_i$ the play is in the vertex 
\[
(q_{i-1},-1,P_{u_{i-1}},P_{u_i},\delta_\mathcal D^*(q_0^\mathcal D,u_0\cdots u_i),\exists).
\]
\end{lemma}

We show the invariant holds by induction.
For $k=0$, the invariant trivially holds.
For the step $k-1 \to k$, assume that Adam's actions have spelt $u_0\cdots u_{k}$, and for $v_0\cdots v_{k-2}$ the invariant holds.
Let $v_{k-1}$ as in \cref{eq:v}, let $\rho_{k-1}$ be a run of the form $q_{k-1} \xrightarrow{{u_{k-1}} / {v_{k-1}}} q_k \in Q_i^\mathcal T$ which must exist by choice of $v_{k-1}$.
The invariant yields that after input $u_0\cdots u_{k-1}$ the play is in the vertex 
\[
  (q_{k-2},-1,P_{u_{k-2}},P_{u_{k-1}},\delta_\mathcal D^*(q_0^\mathcal D,u_0\cdots u_{k-1}),\exists)
\]
 and then Eve takes the producing edge labeled with $(q_{k-2},q_{k-1},c_{k-2})$.
Thus, after further input $u_k$, the play is in the vertex 
\[
  (q_{k-1},-1,P_{u_{k-1}},P_{u_{k}},\delta_\mathcal D^*(q_0^\mathcal D,u_0\cdots u_{k}),\exists).
\]
Since $\rho_{k-1}$ is of the form $\mathcal T\colon q_{k-1}\xrightarrow{u_{k-1} / v_{k-1}} q_k$, the priority $\mathit{best}(P_{u_{k-1}},q_{k-1},q_k)$ is defined, let $c_{k-1}$ denote this priority.
Thus, this vertex has a producing edge with label $(q_{k-1},q_{k},c_{k-1}) \in P_{u_{k-1}}$.
Eve's action is to take this edge.
This concludes the proof that the invariant holds for $k$.

It is left to show the above defined strategy is winning.
Recall the condition for a play $\gamma$ such that it is winning for Eve.
\begin{equation*}
    \mathrm{max}\ \mathrm{Inf}(\mathit{col}_\mathcal D (\gamma)) \text{ is even} \rightarrow \mathrm{max}\ \mathrm{Inf}(\mathit{col}_\mathcal T (\gamma)) \text{ is even}.
    \end{equation*}
Consider a play $\gamma$ according to the above defined strategy, let $\alpha$ be the input word associated to $\gamma$, assume that $\alpha \in \mathrm{dom}(f)$, and let $u_0\cdots$ be its factorization according to \cref{eq:factor}.
By definition of the game graph, the sequence $\mathit{col}_\mathcal D (\gamma)$ is exactly the sequence $c_\mathcal D(\rho_\mathcal D)$, where $\rho_\mathcal D$ is the run of $\mathcal D$ on $\alpha$.
Since $\alpha \in L(\mathcal D)$, we have that $\mathrm{max}\ \mathrm{Inf}(\mathit{col}_\mathcal D (\gamma))$ is even.
Thus, we have to show that $\mathrm{max}\ \mathrm{Inf}(\mathit{col}_\mathcal T (\gamma))$ is even.

Recall that $\mathit{col}_\mathcal T (\gamma)$ is the sequence of colors obtained from the chosen state transformation functions in the play, where $-1$ is the color when no state transformation function is chosen.
Since $\alpha \in \mathrm{dom}(f)$ and $f$ is continuous, the factorization according to \cref{eq:factor} that was used to build the strategy has infinitely many factors.
Hence, infinitely many producing edges are taken, meaning that
$\mathrm{max}\ \mathrm{Inf}(\mathit{col}_\mathcal T (\gamma)) \in C_\mathcal T$.
Furthermore, we also defined infinitely many $v_i$.
Let $\beta = v_0v_1\cdots$, clearly $(\alpha,\beta) \in R$, because $f(\alpha) = \beta$.

Thus, there exists an accepting run of $\mathcal T$ on $\alpha/\beta$.
Consider the run $\rho$ of $\mathcal T$ on $\alpha/\beta$ of the form $\rho_0\rho_1\cdots$ build according to \cref{claim:invariant}.
We show that every accepting run $\rho'$ of $\mathcal T$ on $\alpha/\beta$ has a factorization of the form:

\[
\mathcal T\colon \underbrace{q_0^\mathcal T \xrightarrow{{u_0} / {v'_0}} q_1 \in Q_i^\mathcal T}_{\rho'_0} \quad \underbrace{q_1 \xrightarrow{{u_1} / {v'_1}} q_2\in Q_i^\mathcal T}_{\rho'_1} \quad \cdots \quad \underbrace{q_{i-1} \xrightarrow{{u_{i-1}} / {v'_{i-1}}} q_i\in Q_i^\mathcal T}_{\rho'_{i-1}}.
\]

Assume there is a run $\rho'$ that does not have this property.

\modif{
To begin with, we assume that the run $\rho'$ can not be factorized such that input word $\alpha$ is split into factors $u_0u_1\cdots$.
This problem can not occur, because Item~(1) of property $\mathcal P$ ensures the possiblity of such a factorization. 
Item~(1) ensures that after reading an input letter a state is reached.
Thus, any input factorization can be used to obtain a factorization of the run.
}
\modif{
Hence, assume that the factorization $\rho_0'\rho_1'\cdots$ is picked such that it respects the input factorization, but some reached state is different.
}
Towards a contradiction, pick the first $i$ such that $\rho'_i$ of the form $q_i \xrightarrow{u_i / v_i} p_{i+1}$ with $p_{i+1} \neq q_{i+1}$.
However, since either $v_0\cdots v_i$ is a prefix of $v'_0\cdots v'_i$ or vice versa, because $\mathcal T$ \modif{satisfies Item~(2) of} property $\mathcal P$, it is implied that if $\mathcal T\colon q_i \xrightarrow{u_0\cdots u_i / v_0\cdots v_i} q_{i+1}$ and $\mathcal T\colon q_i \xrightarrow{u_0\cdots u_i / v'_0\cdots v'_i} p_{i+1}$ then $p_{i+1} = q_{i+1}$.
We have a contradiction.

Hence, pick any accepting run $\rho'$ of $\mathcal T$ on $\alpha/\beta$.
Since $\rho'$ is accepting, $\mathrm{max}\ \mathrm{Inf}(c(\rho'))$ is even.
Recall that we have to show that $\mathrm{max}\ \mathrm{Inf}(\mathit{col}_\mathcal T (\gamma))$ is even.
Recall that $\mathit{col}_\mathcal T (\gamma) = c_0c_1\cdots$ according to \cref{claim:invariant}.
We have chosen $c_i$ as $\mathit{best}(P_{u_{i}},q_{i},q_{i+1})$.
Consider the factorization $\rho'_0\rho'_1\cdots$ of $\rho'$ as above.
Clearly, the color $c_i$ is at least as good for Eve as \modif{the
  maximal color that occurs on $\rho'_i$, i.e., $\mathrm{max}\ \text{Occ}(c(\rho'_i))$}.
Thus, since $\mathrm{max}\ \mathrm{Inf}(c(\rho'))$ is even, we also have that $\mathrm{max}\ \mathrm{Inf}(\mathit{col}_\mathcal T (\gamma))$ is even which concludes the proof.
\end{proof}
\section{The Special Case of Closed Domains}
\label{sec:closed}

We turn to the setting of closed domains and show that bounded delay suffices.

\begin{lemma}
  \label[lem]{lemma:closedtobounded}
  Let $R$ with $\mathrm{dom}(R)$ closed be defined by a transducer $\mathcal T$ with property $\mathcal P$.
  If $R$ is uniformizable by a computable function, then there exists
  a computable $\ell \geq 0$ (computable from $\mathcal T$) such that Eve has a winning strategy in $\mathcal G_\mathcal T$ with at most $\ell$ consecutive \emph{skip}-moves.
\end{lemma}

Intuitively, the reason why bounded lookahead suffices in the setting of closed domains is that (basically at each point of time during a play) Adam's moves describe a series of longer and longer finite input words that ``converge'' to a valid infinite input word from the domain.
Hence, Eve can not wait arbitrarily long to make producing moves, as such a play describes a valid infinite input sequence and a finite output sequence. 
We give the formal proof in the upcoming sections.

Together with \cref{thm:functogame,lemma:boundedstrat}, we obtain the following corollary.

\begin{corollary}
  Let $R$ with $\mathrm{dom}(R)$ closed be defined by a transducer $\mathcal T$ with property $\mathcal P$.
  The following are equivalent:
  \begin{enumerate}
    \item $R$ is uniformizable by a computable function.
    \item $R$ is uniformizable by a function computable by a 1DFT.
    \item Eve has a winning strategy in $\mathcal G_\mathcal T$.
  \end{enumerate}

Moreover, if $\mathcal T$ is letter-to-letter, then $R$ is
uniformizable by a computable function iff it is uniformizable by a
function computable by an $\ell$-delay 1DFT for some computable
$\ell$. 
\end{corollary}

We highlight two facts regarding closed domains.

\begin{remark}
The set of infinite words over a finite alphabet is closed, i.e., every total domain is closed.
\end{remark}

\begin{remark}
Furthermore, it is decidable whether a domain (e.g., given by a Büchi automaton) is closed.
\end{remark}

It is a well-known fact that the topological closure of a Büchi language is a Büchi language (one can trim the automaton and declare all states to be accepting) and therefore one can check closedness by checking equivalency with its closure.

\subsection{Uniform Continuity}
\label{app:closedtounif}

In the upcoming proofs, we heavily rely on the notion of uniform continuity.
First, we recall what is meant
by uniform continuity. Intuitively, it means that the amount of input
letters that have to be seen before further output letters can be
determined is independent of the concrete input sequence. Formally, a function $f\colon \Sigma^\omega\rightarrow
\Gamma^\omega$ is \emph{uniformly continuous} if
\begin{equation}\label{eq:continuous-uni}
  \forall i \in \mathbbm N\ \exists j \in \mathbbm N\ \forall \alpha \in \mathrm{dom}(f)\ \forall \beta \in \mathrm{dom}(f)\colon |\alpha \wedge \beta| \geq j \rightarrow |f(\alpha) \wedge f(\beta)| \geq i.
\end{equation}
The function $m : \mathbbm{N}\rightarrow \mathbbm{N}$
which associates to any $i$, some $j$ satisfying the above equation is
called a modulus of continuity. It guarantees that to get $i$ output
symbols, it suffices to read at most $m(i)$ input symbols.

We state two easy to see facts about (uniformly) continuous function used in the remainder.

\begin{remark}
If $f$ is continuous, i.e., $f$ satisfies \cref{eq:continuous}, we have
\begin{equation}\label{eq:fhat-cont}
    \forall u \in \prefs{\mathrm{dom}(f)}\forall i\in\mathbbm N\ \exists j\in\mathbbm N\ \forall u' \in \prefs{\mathrm{dom}(f)}\colon |u \wedge u'| \geq j \rightarrow |\hat{f}(u)| \geq i.
\end{equation}
If $f$ is uniformly continuous, i.e., $f$ satisfies \cref{eq:continuous-uni}, we have
\begin{equation}\label{eq:fhat-uni}
\forall i\in\mathbbm N\ \exists j\in\mathbbm N\ \forall u \in \prefs{\mathrm{dom}(f)}\colon |u| \geq j \rightarrow |\hat{f}(u)| \geq i.
\end{equation}
\end{remark}

To prove \cref{lemma:closedtobounded} we make an important connection between closed domains and uniformly continuous functions.

\modif{It is a well-known fact that the set $\Sigma^\omega$ is a compact
space when equipped with the Cantor distance defined in \cref{note}.}
A closed subset of a compact space is compact, see \cite{arkhangel1990basic}.
Hence, a closed domain $D \subseteq \Sigma^\omega$ is a compact space.
We arrive at the following remark as a consequence of König’s Lemma, or equivalently of the fact that continuous functions on a
compact space are uniformly continuous as stated by the Heine-Cantor
theorem.

\begin{remark}\label[rem]{remark:uniform}
  Assume $R(\mathcal T)$ has a closed domain.
  If $R(\mathcal T)$ is uniformizable by a continuous function, then $R(\mathcal T)$ is uniformizable by a uniformly continuous function.
\end{remark}

\subsection{Finite-state Bounded Delay Game}
\label{page:game-uni}

We show that if $R(\mathcal T)$ is uniformizable by a uniformly continuous function $f$, then Eve has a winning strategy in the corresponding finite-state delay game (based on $\mathcal T$ with property $\mathcal P$) with \emph{bounded delay}.

We show that the bound on the necessary lookahead can be based on their profiles, and provide an adapted game. 
Given a profile $P$, let $L(P)$ denote the set $\{ u \in \Sigma^* \mid P_u = P\}$.
Let $\mathcal P_{fin}$ be the set of profiles whose associated languages are finite, let $L = \bigcup_{P \in \mathcal P_{fin}} L(P)$, and let $\ell$ be the length of the longest word in $L$.
We show that a lookahead of $\ell$ suffices.
We change the game used to model the continuous setting to (implicitly) reflect this bound, i.e., Eve is only allowed to delay to pick skip at most $\ell$ times in a row.

\subparagraph{Finite-state bounded delay game.}
Let $\mathcal G_\mathcal T^{\mathit{uni}} = (G,\mathit{Win})$ be the
infinite-duration turn-based two-player game \modif{obtained by removing from
$\mathcal G_\mathcal T$ (see \cpageref{page:game}) all edges
$\bigl( q,c,P_1,P_2,r,\eve \bigr) \xrightarrow{skip}
\bigl( q,-1,P_1,P_2,r,\adam \bigr)$ such that $L(P_2)$ is infinite}.

\begin{remark}
  The game $\mathcal G_\mathcal T^{\mathit{uni}}$ can be obtained from $\mathcal G_\mathcal T$ in \upshape{\textsc{ExpTime}}.
\end{remark}

To see this one has to realize that for each profile $P$ one has to check whether $|L(P)| \neq \infty$.
The number of profiles is exponential in $\mathcal T$ as is the size of an automaton for $L(P)$.
Checking whether $|L(P)| \neq \infty$ can be done in polynomial time in the automaton for $L(P)$.

\begin{remark}\label[rem]{remark:strat}
  Let $\sigma_\exists$ be a winning strategy for Eve in $\mathcal G_\mathcal T^{\mathit{uni}}$, clearly it is also a winning strategy for Eve in $\mathcal G_\mathcal T$ since $\mathcal G_\mathcal T^{\mathit{uni}}$ is obtained from $\mathcal G_\mathcal T$ by adding edge constraints.
  Furthermore, if Eve plays according to $\sigma_\exists$ she has at most $\ell$ consecutive \emph{skip}-moves.
\end{remark}

\subsection{From Uniformly Continuous Functions to Strategies with Bounded Delay}

\begin{lemma}\label[lem]{lemma:func2strat-uniform}
Let $R$ with $\mathrm{dom}(R)$ closed be defined by a transducer with property $\mathcal P$.
If $R$ is uniformizable by a uniformly continuous function $f$, then Eve has a winning strategy in $\mathcal G_\mathcal T^{\mathit{uni}}$.
\end{lemma}


\begin{proof}
To begin with, we recall some definitions made in the paragraph preceding the definition of $\mathcal G_\mathcal T^{\mathit{uni}}$, see \cpageref{page:game-uni}.
Given a profile $P$, $L(P)$ denotes the set $\{ u \in \Sigma^* \mid P_u = P\}$.
The set $\mathcal P_{fin}$ is the set of profiles whose associated languages are finite, $L = \bigcup_{P \in \mathcal P_{fin}} L(P)$, and $\ell$ is the length of the longest word in $L$.
\modif{Since the game $\mathcal G_\mathcal T^{\mathit{uni}}$ is the game
$\mathcal G_\mathcal T$ where all edges $\bigl( q,c,P_1,P_2,r,\eve
\bigr) \xrightarrow{skip} \bigl( q,-1,P_1,P_2,r,\adam \bigr)$ with
$L(P_2)$ infinite have been removed, Eve never uses $skip$ more than
$\ell$ times in a row.}

We describe how to inductively factorize an input word $\alpha \in \Sigma^\omega$.
Let $u_0 = \varepsilon$, and let $u_0\cdots u_{i-1}$ be the already defined factors of $\alpha$, let $\alpha = u_0\cdots u_{i-1}\alpha'$, and let
\begin{equation}\label{eq:factor-uni}
u_i := \alpha'(:\! k),
\end{equation}
where $k\in\mathbbm N_{\geq 1}$ is the smallest number such that $P_{u_i}\notin \mathcal P_{fin}$, in other words, $L(P_{u_i})$ is infinite.
Note that such a prefix always exists since $L$ is finite.
Moreover, this implies that $|u_i| \leq \ell$ for all $i\in\mathbbm N$.

Let $\alpha_1,\alpha_2 \in \Sigma^\omega$, and let $u_0^1\cdots$ and $u_0^2\cdots$ denote their factorizations built according to \cref{eq:factor-uni}.
Note that $u_i^1 = u_i^2$ for all $i$ such that $|u_0^1\cdots u_i^1| \leq |\alpha_1 \wedge \alpha_2|$.
Furthermore, given an input word $\alpha$, the factorization according to \cref{eq:factor-uni} can be determined on-the-fly reading the input sequence letter-by-letter.

Recall that $f$ uniformizes $R = R(\mathcal T)$, thus, $\mathrm{dom}(f) = \mathrm{dom}(R) = L(\mathcal D)$.
Note that if the prefix ${u_0}\cdots{u_{i}}$ of the input word $\alpha$ from a play can not be completed to a word from $\mathrm{dom}(f)$, we do not care how the strategy behaves from this point on, because then the play automatically belongs to $\mathit{Win}$.
So in the remainder of the proof we only consider plays where Adam spells a word that belongs to $\mathrm{dom}(f)$.

We design a strategy of Eve that behaves as follows.
\modif{For all $i\in\mathbbm N_{\geq 1}$}, after Adam's actions have spelt $u_0\cdots u_i$, Eve's action is to take a producing edge.
In between factors, Eve's action is to skip.
\modif{This ensures that  for all $i\in\mathbbm N_{\geq 1}$, after
Adam's actions have spelt $u_0\cdots u_i$ and Eve has picked an action, the play is in a vertex of the form} 
\[
(q_{i-1},-1,P_{u_{i-1}},P_{u_i},\delta_\mathcal D^*(q_0^\mathcal D,u_0\cdots u_i),\exists)
\]

Note that since $L(P_{u_i})$ is infinite, as ensured by the factorization according to \cref{eq:factor-uni}, Eve cannot use $skip$ and must take a producing edge.
Furthermore, in between factors, Eve is indeed allowed to use the $skip$ action.

Now we explain how we use the sequence of input blocks $u_0u_1\cdots$ of $\alpha$ built as in \cref{eq:factor-uni} to obtain the actions of Eve.
Her actions pick state transformations, our goal is to pick state transformations by defining a sequence of output blocks $v_0v_1\cdots$ and pick state transformations that are induced by $u_0 / v_0 \cdot u_1 / v_1\cdots$.

As described at the beginning of this proof, after a certain amount of input, Eve must make a producing action.
Generally, given two prefixes $u \preceq uu' \in \prefs{\mathrm{dom}(f)}$, it can be the case that $\hat{f}(u) = \hat{f}(uu')$.
However, if $u'$ is such that $L(P_{u'})$ is infinite, then we can
pick an arbitrary long $u''$ that has the same profile as $u'$.
We have that $uu''\in \prefs{\mathrm{dom}(f)}$, because having the same profile means that they induce the same behavior in $\mathcal T$.
Since $f$ is uniformly continuous, $\hat{f}(u) \preceq \hat{f}(uu'')$ if $u''$ is long enough.
The idea is to use $\hat{f}(uu'')$ instead of $\hat{f}(uu')$ to determine Eve's next action.

Formally, in order to choose output blocks $v_0,v_1,\dots$, we inductively build blocks $\bar{u}_0,\bar{u}_1,\dots$, where $\bar{u}_i$ is based on $u_i$ for all $i \in \mathbbm N$.
The idea is that, since each $u_i$ is such that $L(P_{u_i})$ is
infinite for $i\in\mathbbm N_{\geq 1}$, we can pick some $\bar u_i$
that has the same profile as $u_i$ and such that there is a prefix $\bar{v}$ of \modif{the output
\begin{equation}\label{eq:iter-uni}
    \hat{f}(\bar{u}_0\cdots\bar{u}_i)
\end{equation}
that satisfies} $\Delta_\mathcal T^*(q_0^\mathcal
T,\bar{u}_0\cdots\bar{u}_{i-1}/\bar{v}) \cap Q_i^\mathcal T \neq
\emptyset$ ($Q_i^\mathcal T$ denotes the set of input states of
$\mathcal T$ as \modif{defined at the beginning of \cref{sec:completeness}}). This is possible if $\bar{u}_0\cdots\bar{u}_i$ can be completed to an $\omega$-word from $\mathrm{dom}(f)$ because of \cref{eq:fhat-uni}.
This makes it possible to pick output blocks for all previous input blocks.
We let $\bar{v}$ be the shortest prefix of $\hat{f}(\bar{u}_0\cdots\bar{u}_i)$ such that $\Delta_\mathcal T^*(q_0^\mathcal T,\bar{u}_0\cdots\bar{u}_{i-1}/\bar{v}) \cap Q_i^\mathcal T \neq
\emptyset$ and define
\modif{
\begin{equation}\label{eq:barv-uni}
\bar{v}_{i-1} = (\bar{v}_0\cdots\bar{v}_{i-2})^{-1}\bar{v},
\end{equation}
 for all $i \in \mathbbm{N}_{\geq 1}$.}
Based on $\bar u_j$, $\bar v_j$, and $u_j$, we pick $v_j$, which is explained further below.

We show that for the input block sequence $u_0\cdots u_i$, with $\bar u_0\cdots\bar u_{i-1}$ and $\bar v_0\cdots\bar v_{i-1}$ defined as described above, we can pick an output block sequence $v_0\cdots v_{i-1}$ such that the following claim is satisfied for all $i\in\mathbbm N$:

\begin{lemma}\label[lem]{claim:invariant-uni}
If there is a run $\bar\rho$ of $\mathcal T$ on $\bar{u}_0\cdots\bar{u}_{i-1} / \bar{v}_0\cdots\bar{v}_{i-1}$ of the form
\[
\mathcal T\colon \underbrace{q_0^\mathcal T \xrightarrow{\bar{u}_0 / \bar{v}_0} q_1\in Q_i^\mathcal T}_{\bar\rho_0} \quad \underbrace{q_1 \xrightarrow{\bar{u_1} / \bar{v}_1} q_2\in Q_i^\mathcal T}_{\bar\rho_1} \quad \cdots \quad \underbrace{q_{i-1} \xrightarrow{\bar{u}_{i-1} / \bar{v}_{i-1}} q_i\in Q_i^\mathcal T}_{\bar\rho_{i-1}},
\]
then there is a run $\rho$ of $\mathcal T$ on ${u_0}\cdots{u_{i-1}} / {v_0}\cdots{v_{i-1}}$ of the form
\[
\mathcal T\colon \underbrace{q_0^\mathcal T \xrightarrow{{u_0} / {v_0}} q_1}_{\rho_0} \quad \underbrace{q_1 \xrightarrow{{u_1} / {v_1}} q_2}_{\rho_1} \quad \cdots \quad \underbrace{q_{i-1} \xrightarrow{{u_{i-1}} / {v_{i-1}}} q_i}_{\rho_{i-1}},
\]
and the sequence of chosen producing edges in the play after the input $u_0\dots u_i$ has the label sequence 
\[
(q_0^\mathcal T,q_1,c_0)(q_1,q_2,c_1)\cdots(q_{i-1},q_{i},c_{i-1}),
\]
where $c_j = \mathit{best}(P_{u_j},q_j,q_{j+1})$ (see \cref{def:best}) and the edge labeled with $(q_j,q_{j+1},c_{j})$ is taken after the input $u_0\cdots u_{j+1}$ for all $0 \leq j \leq {i-1}$.
Furthermore, after the input $u_0\cdots u_i$ the play is in the vertex 
\[
(q_{i-1},-1,P_{u_{i-1}},P_{u_i},\delta_\mathcal D^*(q_0^\mathcal D,u_0\cdots u_i),\exists).
\]
\end{lemma}

We show the that the claim holds by induction.
For $k = 0$, the claim is trivially true.
For the step $k-1 \to k$, assume that Adam's actions have spelt $u_0\cdots u_k$, we have already defined $\bar{u}_0\cdots \bar{u}_{k-1}$ (which yields the definition of $\bar v_0\cdots \bar{v}_{k-2}$), and $v_0\cdots v_{k-2}$ such that the invariant is satisfied.
Clearly, after $u_0\cdots u_k$, the play is in
\[
(q_{k-1},-1,P_{u_{k-1}},P_{u_k},\delta_\mathcal D^*(q_0^\mathcal D,u_0\cdots u_k),\exists).
\]
We now define $\bar{u}_k$, $\bar{v}_{k-1}$ and $v_{k-1}$ as follows.
Let $\bar u_k$ such that \cref{eq:iter-uni} is satisfied.
Since the profiles of $\bar{u}_0\cdots \bar{u}_{k-1}$ and ${u_0}\cdots {u_{k-1}}$ are the same, it follows that $\bar{u}_0\cdots \bar{u}_{k-1} \in \prefs{\mathrm{dom}(f)}$, because we assumed that ${u_0}\cdots{u_{k-1}} \in \prefs{\mathrm{dom}(f)}$.
This yields the definition of $\bar v_{k-1}$ as given in \cref{eq:barv-uni}.
Pick some $\bar\rho_k$ of the form $\mathcal T\colon q_k \xrightarrow{{\bar u_k} / {\bar v_k}} q_{k+1} \in Q_i^\mathcal T$, and let $c_k = \mathit{best}(P_{\bar{u}_k},q_k,q_{k+1})$, we have $(q_k,q_{k+1},c_k) \in P_{\bar{u}_k}$.
Since $\bar u_k$ and $u_k$ have the same profile, we also have $(q_k,q_{k+1},c_k) \in P_{u_k}$, which implies that we can pick $v_k \in \Gamma^*$ such that $\rho_k$ is of the form $\mathcal T\colon q_k \xrightarrow{{u_k} / {v_k}} q_{k+1}$ and $c_k = \mathrm{max}\ \mathrm{Occ}(c(\rho_k))$.
Furthermore, $(q_k,q_{k+1},c_k) \in P_{u_k}$ implies that Eve can take a producing edge with label $(q_k,q_{k+1},c_k)$.
Hence, the claim is satisfied for $k$.

It is left to argue that the strategy is a winning strategy for Eve.
Recall that if $u_0u_1\cdots \in \mathrm{dom}(f)$, then $\bar{u}_0\bar{u}_1\cdots \in \mathrm{dom}(f)$ and $(\bar{u}_0\bar{u}_1\cdots,f(\bar{u}_0\bar{u}_1\cdots)) \in R$.
The same argumentation as in the proof of \cref{lemma:func2strat} yields that the strategy is winning.
The main point is to prove that every accepting run $\bar{\rho}'$ of $\mathcal T$ on $\bar{u}_0\bar{u}_1\cdots/f(\bar{u}_0\bar{u}_1\cdots)$ can be factorized into the form 
\[
\mathcal T\colon \underbrace{q_0^\mathcal T \xrightarrow{{\bar{u}_0} / {\bar{v}'_0}} q_1}_{\bar{\rho}'_0} \quad \underbrace{q_1 \xrightarrow{{\bar{u}_1} / {\bar{v}'_1}} q_2}_{\bar{\rho}'_1} \quad \cdots \quad \underbrace{q_{i-1} \xrightarrow{{\bar{u}_{i-1}} / {\bar{v}'_{i-1}}} q_i}_{\bar{\rho}'_{i-1}},
\]
with the same intermediate states as both the runs $\rho$ and $\bar{\rho}$ inductively constructed in order to define the strategy that satisfies \cref{claim:invariant-uni}.
In order to prove that this factorization of $\bar{\rho}'$ exists, we use that $\mathcal T$ has property $\mathcal P$ exactly as we have done in the proof of \cref{thm:functogame}.
Since $\bar{\rho}'$ is accepting and each $c_i$ is at least as good as
$\mathrm{max}\ \mathrm{Occ}(c(\bar{\rho}'_i))$, it is easy to see that the strategy is winning.
\end{proof}

Finally, we are able to prove \cref{lemma:closedtobounded}.
The statement is a direct consequence of \cref{remark:uniform} together with \cref{lemma:func2strat-uniform} and \cref{remark:strat}.

\section{Discussion}
\label{sec:discussion}

\subparagraph{Continuous functions.} 
We have shown that
checking the existence of a computable function uniformizing a relation given by transducer with property $\mathcal P$ is decidable (a consequence of \cref{thm:gametofunc,thm:functogame,lemma:complexity}).
The proofs of \cref{thm:gametofunc,thm:functogame} use another notion, which is easier
to manipulate mathematically than computability, namely that of continuity.
These notions are closely related.
If a function $f\colon \Sigma^\omega \to \Gamma^\omega$ is computable, it is also continuous.
This is not difficult to see when comparing the definitions of
computable and continuous functions. The converse does not hold
because the continuity definition does not have any computability
requirements (see \cite{DBLP:conf/concur/DaveFKL20} for a counter-example). However, regarding synthesis, the two notions coincide:

\begin{theorem}\label{thm:equiv}\label{thm:main}
  Let $R$ be defined by $\mathcal T$ with property $\mathcal P$.
  The following are equivalent:
  \begin{enumerate}
  \item $R$ is uniformizable by a continuous function.
  \item $R$ is uniformizable by a computable function.
\end{enumerate}
\end{theorem}

\begin{proof}
  Indeed, any computable uniformizer is continuous.
  \cref{thm:functogame} states that if there exists a computable uniformizer, then there exists a winning strategy in the delay game.
  However, in the proof of this theorem, we show a stronger statement:
 If there exists a continuous uniformizer, then there exists a winning strategy in the delay game.
 Such a strategy can be assumed to have finite-memory (as
finite-memory suffices to win games with $\omega$-regular
conditions). We have shown in the proof of \cref{lemma:strat2func} how
to translate a finite-state winning strategy into an algorithm (a
Turing machine) that computes a function~$f$ which uniformizes the relation.
\end{proof}

\subparagraph{Further undecidability results.}

\cref{thm:undecidable} states that it is undecidable whether a rational relation is uniformizable by a computable function.
As we have seen in \cref*{sec:completeness}, restricting the class of specifications (to \DRAT) yields decidability.
Another approach to obtain decidability is to change the class of desired implementations, and not the class of specifications.
The next proposition shows that this is not fruitful.

\begin{proposition}\label[prop]{thm:undecextra}
It is undecidable whether a rational relation $R$ is uniformizable by a sequential resp.\ letter-to-letter sequential function, even if $R$ has total domain.
\end{proposition}

\begin{proof}
To prove \cref{thm:undecidable}, we sketched a reduction from Post's correspondence problem.
The reduction from Post's correspondence problem yields the undecidability for computable as well as sequential functions but not for letter-to-letter sequential functions.
Hence, we now give a slightly more complicated proof that also yields
the desired result for letter-to-letter sequential functions. \modif{Similar ideas have been used in~\cite{CarayolL14} to prove that
  uniformizability of rational functions by sequential functions of finite words is undecidable}.

We show all undecidability results by reduction from the halting problem for Turing machines.
Let $M$ be a deterministic Turing machine, and let $R_M$ be the relation that contains pairs $(\alpha,\beta)$ of the form $\alpha = c_0\$c_1\$\cdots c_n\$\alpha'$ and $\beta = c'_0\$c'_1\$\cdots c'_n\$\beta'$, where $c_0,\dots,c_n,c'_0,\dots,c'_n$ code configurations of $M$ in the usual way, $c_0$ is the initial configuration, $c_n$ is a halting configuration, $\alpha'$ and $\beta' \in \{a,b\}^\omega$, and either $\beta = \alpha$ if $\alpha'$ contains infinitely many $a$, or there is some $i \in \{0,\dots,n-1\}$ such that $\mathit{succ}(c_i) \neq c'_{i+1}$, where $\mathit{succ}(c_i)$ is the successor configuration of $c_i$ if $\alpha'$ contains finitely many $a$.
We call inputs of the form $c_0\$c_1\$\cdots c_n\$\alpha'$ a valid encoding.

We make the relation $R_M$ total by adding all pairs $(\alpha,\beta)$, where $\alpha$ is not a valid encoding (we do not care about the outputs then).

We argue that $R_M$ is $\omega$-rational.
Clearly, a transducer can guess whether the input is a valid encoding and verify this.
If the input is not a valid encoding, the transducer must find an encoding error, and has nothing to verify regarding the output.

If the input is a valid encoding, the transducer can guess and verify whether the input contains infinitely many $a$.
If it guesses infinitely many $a$, it has to verify that input and output are equal.
If it guesses finitely many $a$, it reads the first configuration of the output, and can then read $c_i$ and $c'_{i+1}$ in parallel for all $i \in \{0,\cdots,n-1\}$ and check whether $\mathit{succ}(c_i)\neq c'_{i+1}$. 

We show that $R_M$ is uniformizable by a continuous function, a computable function, a sequential function, resp.\ a letter-to-letter sequential function iff $M$ does not halt on the empty tape.

Assume $M$ does not halt.
Let $f_{\mathit{id}}$ be the identity function.
Clearly, $f_{\mathit{id}}$ is continuous, computable, sequential, and synchronous sequential.
It is easy to verify that $f_{\mathit{id}}$ uniformizes $R_M$.
Consider an infinite input word $\alpha$.
We only consider the interesting case that $\alpha$ is a valid encoding.
If $\alpha$ contains infinitely many $a$, clearly $(\alpha,f_{\mathit{id}}(\alpha)) \in R_M$.
If $\alpha$ contains finitely many $a$, there has to be some $i \in \{0,\dots,n-1\}$ such that $\mathit{succ}(c_i) \neq c'_{i+1}$.
This is true, because $M$ does not halt, thus, the configuration sequence $c_0\$c_1\$\cdots c_n$ has an error, i.e., there is some $i$ such that $\mathit{succ}(c_i) \neq c_{i+1}$.
Hence, since $c_0\$c_1\$\cdots c_n\$ = c'_0\$c'_1\$\cdots c'_n\$$, we obtain $\mathit{succ}(c_i) \neq c'_{i+1}$.

Assume $M$ does halt.
Towards a contradiction, assume that $f$ is a function that uniformizes $R_M$ that is either continuous, computable, sequential, or synchronous sequential.
Note that every synchronous sequential function is sequential, every sequential function is computable and every computable function is continuous.
We arrive at a contradiction by using only properties of continuous functions.
Since $M$ halts, there exists a configuration sequence starting with the initial configuration that ends with a halting configuration that has no errors.
Let $u$ denote this configuration sequence.
Consider $ua^\omega \in \mathrm{dom}(f)$.
According to \cref{eq:continuous}, there exists some $j \in \mathbbm N$ such that for all $\beta \in \mathrm{dom}(f)$ holds that $|ua^\omega \wedge \beta| \geq j$ implies that $|f(ua^\omega) \wedge f(\beta)| \geq |u|$. 
Consider $ua^jb^\omega \in \mathrm{dom}(f)$, we have $f(ua^\omega) = ua^\omega$ and $f(ua^jb^\omega) = u'\beta'$ with $|u| = |u'|$ and $u \neq u'$, because $u'$ must begin with a configuration sequences such that there is some $i \in \{0,\dots,n-1\}$ such that $\mathit{succ}(c_i) \neq c'_{i+1}$.
Since there is no $i$ such that $\mathit{succ}(c_i) \neq c_{i+1}$,
because $u$ has no errors, this implies that $u \neq u'$. Thus, $|ua^\omega \wedge ua^jb^\omega| \geq j$, but $|f(ua^\omega) \wedge f(ua^jb^\omega)| < |u|$, which is a contradiction.
\end{proof}

\section{Conclusion and Future Work}

We investigated the synthesis of algorithms (a.k.a.\ Turing machine computable functions) from rational specifications.
While undecidable in general, we have proven decidability for \DRAT (and a fortiori \AUT).
Furthermore, we have shown that the whole computation power of Turing
machines is not needed, two-way transducers are sufficient (and
necessary). As TMs reading heads are read-only left-to-right,
converting a 2DFT into a TM requires that the TM stores longer and
longer prefixes of the input in the working-tape for later
access. This is the only use the TM needs to make of the working tape. 
This is a naive translation, and sometimes the working tape can
be flushed (some prefixes of the input may possibly not be needed
anymore). More generally, it is an interesting research direction to
fine-tune the class of functions targeted by synthesis with respect to
some constraints on the memory, including quantitative
constraints.

Related to the latter research direction is the following open
question: is the synthesis problem of functions computable by
input-deterministic one-way (a.k.a.\ sequential) transducers
from deterministic rational relations decidable? It is already open
for automatic relations. We have shown that if a rational relation
with closed domain is uniformizable by a computable function, then it
also uniformizable by a sequential function.
However, closedness is not a \modif{necessary} condition: e.g., the function which maps any $a^nxc^\omega$ to $xc^\omega$ for $x\in \{\#,\$\}$, is sequential; a sequential transducer just has to erase the $a^n$ part, but its domain is not closed.
 This problem is interesting because sequential transducers only
require bounded memory to compute a function (in contrast to two-way
transducers that require access to unboundedly large prefixes of the
input).

\paragraph{Acknowledgments.} \modif{We warmly thank the anonymous reviewers for their
helpful comments and in particular some reviewer for spotting an error
in the preliminary version of this article, that led to redefining
property $\mathcal{P}$ of \cref{sec:completeness}.}


\bibliographystyle{plain}

\bibliography{biblio}

\begin{thebibliography}{10}

\bibitem{almagor2020good}
Shaull Almagor and Orna Kupferman.
\newblock Good-enough synthesis.
\newblock In {\em International Conference on Computer Aided Verification}, pages 541--563. Springer, 2020.

\bibitem{DBLP:series/natosec/AlurBDF0JKMMRSSSSTU15}
Rajeev Alur, Rastislav Bod{\'{\i}}k, Eric Dallal, Dana Fisman, Pranav Garg, Garvit Juniwal, Hadas Kress{-}Gazit, P.~Madhusudan, Milo M.~K. Martin, Mukund Raghothaman, Shambwaditya Saha, Sanjit~A. Seshia, Rishabh Singh, Armando Solar{-}Lezama, Emina Torlak, and Abhishek Udupa.
\newblock Syntax-guided synthesis.
\newblock In {\em Dependable Software Systems Engineering}, pages 1--25. IEEE, 2015.

\bibitem{alur2012regular}
Rajeev Alur, Emmanuel Filiot, and Ashutosh Trivedi.
\newblock Regular transformations of infinite strings.
\newblock In {\em 2012 27th Annual IEEE Symposium on Logic in Computer Science}, pages 65--74. IEEE, 2012.

\bibitem{arkhangel1990basic}
AV~Arkhangel’skiǐ and Vitaly~V Fedorchuk.
\newblock The basic concepts and constructions of general topology.
\newblock In {\em General Topology I}, pages 1--90. Springer, 1990.

\bibitem{DBLP:books/lib/Berstel79}
Jean Berstel.
\newblock {\em Transductions and context-free languages}, volume~38 of {\em Teubner Studienb{\"{u}}cher : Informatik}.
\newblock Teubner, 1979.

\bibitem{DBLP:conf/atva/BloemEK15}
Roderick Bloem, R{\"{u}}diger Ehlers, and Robert K{\"{o}}nighofer.
\newblock Cooperative reactive synthesis.
\newblock In Bernd Finkbeiner, Geguang Pu, and Lijun Zhang, editors, {\em Automated Technology for Verification and Analysis - 13th International Symposium, {ATVA} 2015, Shanghai, China, October 12-15, 2015, Proceedings}, volume 9364 of {\em Lecture Notes in Computer Science}, pages 394--410. Springer, 2015.

\bibitem{DBLP:journals/acta/BrenguierRS17}
Romain Brenguier, Jean{-}Fran{\c{c}}ois Raskin, and Ocan Sankur.
\newblock Assume-admissible synthesis.
\newblock {\em Acta Informatica}, 54(1):41--83, 2017.

\bibitem{BuLa69}
J.~Richard B{\"u}chi and Lawrence~H. Landweber.
\newblock Solving sequential conditions by finite-state strategies.
\newblock {\em Trans. Ameri. Math. Soc.}, 138:295--311, 1969.

\bibitem{CarayolL14}
Arnaud Carayol and Christof L{\"o}ding.
\newblock {Uniformization in Automata Theory}.
\newblock In {\em Proceedings of the 14th Congress of Logic, Methodology and Philosophy of Science Nancy, July 19-26, 2011}, pages 153--178. London: College Publications, 2014.

\bibitem{DBLP:conf/mfcs/CartonD22}
Olivier Carton and Ga{\"{e}}tan Dou{\'{e}}neau{-}Tabot.
\newblock Continuous rational functions are deterministic regular.
\newblock In {\em {MFCS}}, volume 241 of {\em LIPIcs}, pages 28:1--28:13. Schloss Dagstuhl - Leibniz-Zentrum f{\"{u}}r Informatik, 2022.

\bibitem{DBLP:conf/tacas/ChatterjeeH07}
Krishnendu Chatterjee and Thomas~A. Henzinger.
\newblock Assume-guarantee synthesis.
\newblock In Orna Grumberg and Michael Huth, editors, {\em Tools and Algorithms for the Construction and Analysis of Systems, 13th International Conference, {TACAS} 2007, Held as Part of the Joint European Conferences on Theory and Practice of Software, {ETAPS} 2007 Braga, Portugal, March 24 - April 1, 2007, Proceedings}, volume 4424 of {\em Lecture Notes in Computer Science}, pages 261--275. Springer, 2007.

\bibitem{choffrut1999uniformization}
Christian Choffrut and Serge Grigorieff.
\newblock Uniformization of rational relations.
\newblock In {\em Jewels are Forever}, pages 59--71. Springer, 1999.

\bibitem{clarke2018handbook}
Edmund~M Clarke, Thomas~A Henzinger, Helmut Veith, and Roderick Bloem.
\newblock {\em Handbook of model checking}, volume~10.
\newblock Springer, 2018.

\bibitem{DBLP:conf/icalp/ConduracheFGR16}
Rodica Condurache, Emmanuel Filiot, Raffaella Gentilini, and Jean{-}Fran{\c{c}}ois Raskin.
\newblock The complexity of rational synthesis.
\newblock In Ioannis Chatzigiannakis, Michael Mitzenmacher, Yuval Rabani, and Davide Sangiorgi, editors, {\em 43rd International Colloquium on Automata, Languages, and Programming, {ICALP} 2016, July 11-15, 2016, Rome, Italy}, volume~55 of {\em LIPIcs}, pages 121:1--121:15. Schloss Dagstuhl - Leibniz-Zentrum f{\"{u}}r Informatik, 2016.

\bibitem{DBLP:conf/concur/DaveFKL20}
Vrunda Dave, Emmanuel Filiot, Shankara~Narayanan Krishna, and Nathan Lhote.
\newblock Synthesis of computable regular functions of infinite words.
\newblock In {\em {CONCUR}}, volume 171 of {\em LIPIcs}, pages 43:1--43:17. Schloss Dagstuhl - Leibniz-Zentrum f{\"{u}}r Informatik, 2020.

\bibitem{engelfriet2001mso}
Joost Engelfriet and Hendrik~Jan Hoogeboom.
\newblock Mso definable string transductions and two-way finite-state transducers.
\newblock {\em ACM Transactions on Computational Logic (TOCL)}, 2(2):216--254, 2001.

\bibitem{DBLP:conf/icla/Filiot15}
Emmanuel Filiot.
\newblock Logic-automata connections for transformations.
\newblock In Mohua Banerjee and Shankara~Narayanan Krishna, editors, {\em Logic and Its Applications - 6th Indian Conference, {ICLA} 2015, Mumbai, India, January 8-10, 2015. Proceedings}, volume 8923 of {\em Lecture Notes in Computer Science}, pages 30--57. Springer, 2015.

\bibitem{DBLP:conf/icalp/FiliotJLW16}
Emmanuel Filiot, Isma{\"{e}}l Jecker, Christof L{\"{o}}ding, and Sarah Winter.
\newblock On equivalence and uniformisation problems for finite transducers.
\newblock In {\em {ICALP}}, volume~55 of {\em LIPIcs}, pages 125:1--125:14. Schloss Dagstuhl - Leibniz-Zentrum f{\"{u}}r Informatik, 2016.

\bibitem{DBLP:conf/fsttcs/FiliotW21}
Emmanuel Filiot and Sarah Winter.
\newblock Synthesizing computable functions from rational specifications over infinite words.
\newblock In {\em {FSTTCS}}, volume 213 of {\em LIPIcs}, pages 43:1--43:16. Schloss Dagstuhl - Leibniz-Zentrum f{\"{u}}r Informatik, 2021.

\bibitem{DBLP:journals/sttt/FinkbeinerS13}
Bernd Finkbeiner and Sven Schewe.
\newblock Bounded synthesis.
\newblock {\em Int. J. Softw. Tools Technol. Transf.}, 15(5-6):519--539, 2013.

\bibitem{DBLP:conf/atva/GerstackerKF18}
Carsten Gerstacker, Felix Klein, and Bernd Finkbeiner.
\newblock Bounded synthesis of reactive programs.
\newblock In Shuvendu~K. Lahiri and Chao Wang, editors, {\em Automated Technology for Verification and Analysis - 16th International Symposium, {ATVA} 2018, Los Angeles, CA, USA, October 7-10, 2018, Proceedings}, volume 11138 of {\em Lecture Notes in Computer Science}, pages 441--457. Springer, 2018.

\bibitem{gradel2002automata}
Erich Gradel, Wolfgang Thomas, and Thomas Wilke.
\newblock {\em Automata, logics, and infinite games: a guide to current research}, volume 6014.
\newblock Springer Science \& Business Media, 2002.

\bibitem{holtmann2010degrees}
Michael Holtmann, {\L}ukasz Kaiser, and Wolfgang Thomas.
\newblock Degrees of lookahead in regular infinite games.
\newblock In {\em International Conference on Foundations of Software Science and Computational Structures}, pages 252--266. Springer, 2010.

\bibitem{hopcroft1967approach}
John~E Hopcroft and Jeffrey~D Ullman.
\newblock An approach to a unified theory of automata.
\newblock {\em The Bell System Technical Journal}, 46(8):1793--1829, 1967.

\bibitem{hosch1972finite}
F~Hosch and Lawrence Landweber.
\newblock Finite delay solutions for sequential conditions.
\newblock Technical report, University of Wisconsin-Madison Department of Computer Sciences, 1972.

\bibitem{DBLP:journals/corr/KleinZ14}
Felix Klein and Martin Zimmermann.
\newblock How much lookahead is needed to win infinite games?
\newblock {\em Log. Methods Comput. Sci.}, 12(3), 2016.

\bibitem{IC::Kobayashi69}
K.~Kobayashi.
\newblock Classification of formal languages by functional binary transductions.
\newblock {\em Information and Control}, 15(1):95--109, July 1969.

\bibitem{DBLP:journals/amai/KupfermanPV16}
Orna Kupferman, Giuseppe Perelli, and Moshe~Y. Vardi.
\newblock Synthesis with rational environments.
\newblock {\em Ann. Math. Artif. Intell.}, 78(1):3--20, 2016.

\bibitem{PnuRos:89}
A.~Pnueli and R.~Rosner.
\newblock On the synthesis of a reactive module.
\newblock In {\em ACM Symposium on Principles of Programming Languages (POPL)}. ACM, 1989.

\bibitem{safra1992exponential}
Shmuel Safra.
\newblock Exponential determinization for $\omega$-automata with strong-fairness acceptance condition.
\newblock In {\em Proceedings of the twenty-fourth annual ACM symposium on Theory of computing}, pages 275--282, 1992.

\bibitem{sakarovitch2009elements}
Jacques Sakarovitch.
\newblock {\em Elements of automata theory}.
\newblock Cambridge University Press, 2009.

\bibitem{schewe2007solving}
Sven Schewe.
\newblock Solving parity games in big steps.
\newblock In {\em International Conference on Foundations of Software Technology and Theoretical Computer Science}, pages 449--460. Springer, 2007.

\bibitem{weihrauch2012computable}
Klaus Weihrauch.
\newblock {\em Computable analysis: an introduction}.
\newblock Springer Science \& Business Media, 2012.

\end{thebibliography}

\end{document}